\documentclass[11pt]{article}

\usepackage{amssymb, amsmath, enumerate, paralist}

\usepackage{graphicx,picins}

\usepackage{algorithmicx}
\usepackage[noend]{algpseudocode}
\usepackage{amsmath,amsfonts,mathrsfs}
\usepackage{xspace}
\usepackage{graphicx}
\usepackage{epsfig}
\usepackage{authblk}
\usepackage{algorithm}
\usepackage{makeidx}  
\usepackage{hyperref}

\usepackage{tikz}
\usetikzlibrary{backgrounds}
\usetikzlibrary{arrows}

 \tikzstyle{Gcentre}=[rectangle,draw,fill=black!90,minimum size=4.5pt, inner sep=0pt]
 \tikzstyle{Filtered}=[rectangle,draw,fill=black!20,minimum size=4.5pt, inner sep=0pt]
 \tikzstyle{Lcentre}=[circle,draw=black,fill=black!90,minimum size=4.5pt,inner sep=0pt]
 \tikzstyle{LFiltered}=[circle,draw=black,fill=black!20,minimum size=4.5pt,inner sep=0pt]

 \tikzstyle{Point}=[circle,draw=black,fill=blue!90,minimum size=4pt,inner sep=0pt]

 \tikzstyle{Lcentre-tiny}=[circle,draw=black,fill=black!90,minimum size=1pt,inner sep=0pt]

\newtheorem{theorem}{Theorem} 
\newtheorem{lemma}{Lemma} 
\newtheorem{claim}{Claim}

\newtheorem{definition}{Definition} 
\newtheorem{remarka}{Remark}

\newenvironment{proof}{{\bf Proof.}}{\hfill\rule{2mm}{2mm}} 

\newenvironment{pproof}[1]{\noindent{\textbf{Proof of #1.}}}{\hfill\rule{2mm}{2mm}} 

\newcommand{\RR}{\mathbb{R}}

\newcommand{\cl}{{\mathcal C}\xspace}

\newcommand{\cost}{{\rm cost}\xspace}

\newcommand{\cX}{{\mathcal X}\xspace}
\newcommand{\fa}{{\mathcal C}\xspace}
\newcommand{\opt}{{\mathcal O}\xspace}
\newcommand{\loc}{{\mathcal S}\xspace}

\newcommand{\cents}{{\mathcal T}\xspace}
\newcommand{\net}{{\mathcal N}\xspace}
\newcommand{\eps}{{\epsilon}}

\newcommand{\calS}{{\cal S}}
\newcommand{\del}{\delta}

\newcommand{\kmeans}{{$k$-\textsc{means}}\xspace}
\newcommand{\kmed}{{$k$-\textsc{median}}\xspace}
\newcommand{\uufl}{{\textsc{uniform-cost UFL}}\xspace}
\newcommand{\ufl}{{\textsc{uncapacitated facility location}}\xspace}

\newcommand{\kmeanso}{{$k$-\textsc{means}-\textsc{out}}\xspace}
\newcommand{\kmedo}{{$k$-\textsc{median}-\textsc{out}}\xspace}
\newcommand{\uflo}{{\textsc{UFL}-\textsc{out}}\xspace}
\newcommand{\uuflo}{{\textsc{Uniform-UFL}-\textsc{out}}\xspace}

\newcommand{\cP}{{\mathcal P}\xspace}
\newcommand{\cG}{{\mathcal G}\xspace}

\title{Approximation Schemes for Clustering with Outliers}
\author{Zachary Friggstad\thanks{This research was undertaken, in part, thanks to funding from the Canada Research Chairs program and an NSERC Discovery Grant.}
\quad Kamyar Khodamoradi\\
\quad Mohsen Rezapour
\quad Mohammad R. Salavatipour\thanks{Supported by NSERC.}\\
\vspace{1cm}
Department of Computing Science\\
University of Alberta}

\date{}
\begin{document}

\maketitle
\thispagestyle{empty}
\begin{abstract}
Clustering problems are well-studied in a variety of fields such as data science, operations research, and computer science. 
Such problems include variants of centre location problems, $k$-median, and $k$-means to name a few.
In some cases, not all data points need to be clustered; some may be discarded for various reasons. 
For instance, some points may arise from noise in a data set or one might be willing to discard a certain fraction of 
the points to avoid incurring unnecessary overhead in the cost of a clustering solution.

We study clustering problems with outliers. More specifically, we look at \ufl (UFL), \kmed, and \kmeans.
In these problems, we are given a set $\cX$ of data points in a metric space $\delta(.,.)$, 
a set $\fa$ of possible centres (each maybe with an opening
cost), maybe an integer parameter $k$, plus an additional parameter $z$ as the number of outliers.
In \ufl with outliers, we have to open some centres, discard up to $z$ points of $\cX$ and assign every other point to the nearest
open centre, minimizing the total assignment cost plus centre opening costs. In \kmed and \kmeans, 
we have to open up to $k$ centres but there are no opening costs. In \kmeans, the cost of assigning $j$ to $i$ is $\delta^2(j,i)$.
We present several results. Our main focus is on cases 
where $\delta$ is a doubling metric (this includes fixed dimensional Euclidean metrics as a special case)
or is the shortest path metrics of graphs from a minor-closed family of graphs.
For \uufl with outliers on such metrics we show that a multiswap simple  local search heuristic yields a PTAS.
With a bit more work, we extend this to bicriteria approximations for the \kmed and \kmeans problems in the same metrics 
where, for any constant $\epsilon > 0$, we can find a solution using $(1+\epsilon)k$ centres whose cost is at most a 
$(1+\epsilon)$-factor of the optimum and uses at most $z$ outliers.
Our algorithms are all based on natural multiswap local search heuristics. 
We also show that natural local search heuristics that do not violate the number of clusters and outliers for \kmed (or \kmeans)
will have unbounded gap even in Euclidean metrics.

Furthermore, we show how our analysis can be extended to general metrics for \kmeans 
with outliers to obtain a $(25+\epsilon,1+\epsilon)$
bicriteria: an algorithm that uses at most $(1+\epsilon)k$ clusters and whose cost is at most $25+\epsilon$ of optimum and uses
no more than $z$ outliers.

\end{abstract}
\newpage
\setcounter{page}{1}
\section{Introduction}

Clustering is a fundamental problem in the field of data analysis with a long history and a wide range of applications in very
 different areas, including data mining \cite{ber06}, image processing \cite{PT92}, biology \cite{Jain10}, and database systems 
\cite{EKSX96}. Clustering is the task of  partitioning a given set of data points into clusters based on a specified
similarity measure between the data points such that the points within the same cluster are more similar to each other than those 
in different clusters. 

In a typical clustering problem, we are given
a set of $n$ data points in a metric space, and an integer $k$ which specifies the desired number of clusters. We wish to find a 
set of $k$ points to act as {\em centres} and then assign each point to its nearest centre, thereby forming $k$ clusters. The 
quality of the clustering solution can be measured by using different objectives. For example, in the \kmeans clustering (which 
is the most widely used clustering model), the goal (objective function) is to minimize the sum of squared distances of each data 
point to its 
centre, while in \kmed, the goal is to minimize the sum of distances of each data point to its 
centre. The \ufl problem is the same as \kmed except that instead of a cardinality constraint bounding the number of centres, there 
is an additional cost for each centre included in the solution.
Minimizing these objective functions exactly is NP-hard \cite{ADHP09,DFKVV04,GK99,JMS02,MNV09,Vattani09}, so there has been
 substantial work on obtaining provable upper
 bounds (approximability) and  lower bounds (inapproximability)  for these objectives; see \cite{ANSW16,BPRS14,GK99,JMS02,LSW17,L11}
 for the currently best bounds.
Although inapproximability results \cite{GK99,JMS02,LSW17} prevent getting polynomial time approximation schemes (PTASs) for these
 problems in general metrics, PTASs are known for these
 problems in fixed dimensional Euclidean metrics \cite{ARR98,AKM16B,FRS16B}.
Indeed, PTASs for \kmed and \ufl in fixed dimension Euclidean space \cite{ARR98} have been known for almost two
decades, but getting a PTAS for \kmeans in fixed dimension Euclidean space had been an open problem until recent results
 of \cite{AKM16B,FRS16B}.

In spite of the fact that these popular (centre based) clustering models are reasonably good for noise-free data sets, their 
objective functions (specially the \kmeans objective function) are
 extremely sensitive to the existence of points far from cluster centres. Therefore, a small number of very distant data points,
 called {\em outliers}, --if not discarded-- can
  dramatically affect the clustering cost and also the quality of the final clustering solution.
Dealing with such outliers is indeed the main focus of this paper. Clustering with outliers has a natural motivation in many 
applications of clustering and
centre location problems. For example, consider (nongovernmental) companies that provide nationwide services (e.g., mobile phone
 companies, chain stores). They alway have to disregard some percentage of the remote population in order to be able to provide a
 competitive service to the majority of the population.

We restrict our attention to the outlier version of the three well studied clustering problems: \kmeans with outliers 
(\kmeanso), \kmed with outliers (\kmedo), and \ufl with outliers (\uflo).
Formally, in these problems, we are given a set $\cX$ of $n$ data points in a metric space, a set $\fa$ of possible centres, and 
the number of desired outliers $z$.
Both \kmeanso and \kmedo aim at finding $k$ centres $C=\{c_1,\ldots,c_k \} \subseteq \fa$ and a set of (up to) $z$ points $Z$ 
to act as outliers.
The objective is to minimize the clustering cost. In \kmeanso, this is the sum of squared distances of each data point in
 $\cX \setminus Z$ to its nearest 
centre, i.e., $\sum_{x\in \cX \setminus Z}\delta(x,C)^2$, while in \kmedo this is just the sum of distances, i.e.,
 $\sum_{x\in \cX \setminus Z}\delta(x,C)$, where 
$\delta(x,C)$ indicates the distance of point $x$ to its nearest centre in $C$. \uflo is the same as \kmedo except that instead of a
cardinality constraint, we are given opening cost $f_c$ for each centre $c \in \fa$. The problem hence consists of finding 
centres (facilities) $C$ and $z$ outliers $Z$ that minimizes $\sum_{x\in \cX \setminus Z}\delta(x,C) + \sum_{c \in C} f_c$.
We present PTASs for these problems on doubling metrics i.e. metrics with fixed doubling dimensions  (which include
fixed dimension Euclidean metrics as special case) and shortest path metrics of
minor closed graphs (which includes planar graphs as special case)\footnote{For brevity, we will call such graphs {\em minor closed}, understanding this means they belong to a fixed family of graphs that closed under minors.}.
Recall that a metric $(V,\delta)$ has doubling dimension $d$ 
if each ball of radius $2r$ can be covered with $2^d$ balls of radius $r$ in $V$. We call it a {\em doubling metric} if $d$ can be regarded as a constant; Euclidean metrics of constant (Euclidean) dimension are doubling metrics.

Despite a very large amount of work on the clustering problems, there has been only little work on their outlier versions.
To the best of our knowledge, the clustering problem with outliers was introduced by Charikar et al. \cite{CKMN01}.
They devised a factor 3-approximation for \uflo, and also a bicriteria $4(1+\frac{1}{\epsilon})$-approximation algorithm 
for \kmedo that drops $z(1+ \epsilon)$ outliers. 
They obtained these results via some modifications of the Jain-Vazirani algorithm \cite{JV01}.
The first true approximation algorithm for \kmedo was given by Chen \cite{Chen08} who obtained this by combining very 
carefully the Jain-Vazirani algorithm and local search; the approximation factor is not specified but seams to be a very large constant.  
Very recently, the first  bicriteria approximation
 algorithm for \kmeanso is obtained by Gupta et al. \cite{GKLMV17}. They devised a bicriteria $274$-approximation
  algorithm for \kmeanso that drops $O(kz\log(n\Delta))$ outliers, where $\Delta$ denotes the maximum distance between data points. 
This is obtained by a simple local search heuristic for the problem. 

\subsection{Related work}
\kmeans is one of the most widely studied problems in the Computer Science literature. The problem is usually considered on
 $d$-dimensional Euclidean space $\RR^d$, where the objective becomes minimizing the variance of the data points with respect
 to the centres they are assigned to.
The most commonly used algorithm for \kmeans is a simple heuristic known as Lloyd's algorithm (commonly referred to as the
 \kmeans algorithm) \cite{Lloyd82}.
Although this algorithm works well in practice, it is known that the 
the cost of the solutions computed by this algorithm can be arbitrarily large compared to the optimum solution \cite{KMNPSW04}.
Under some additional assumptions about the initially chosen centres, however, 
Arthur and Vassilvitskii~\cite{AV07} show that the approximation ratio of Lloyd's algorithm is $O(\log k)$. Later, Ostrovsky 
et al.~\cite{ORSS12} show that  the 
approximation ratio is bounded by a constant if the input points obey some special properties.
Under no such assumptions, Kanungo et al.~\cite{KMNPSW04} proved that a simple local search heuristic (that swaps only  a
 constant number of centres in each iteration) yields an $(9+\epsilon)$-approximation algorithm for Euclidean \kmeans. Recently, 
Ahmadian et al. \cite{ANSW16} improved the approximation ratio to $6.357 + \epsilon$ by primal-dual algorithms. 
For general metrics, Gupta and Tangwongsan \cite{GT08} proved that the local
search algorithm is a $(25+\eps)$-approximation. This was also recently improved to $9+\epsilon$ via primal-dual algorithms 
\cite{ANSW16}.

In order to obtain algorithms with arbitrary small approximation ratios for Euclidean \kmeans, many researchers restrict their
 focus on cases when $k$ or $d$ is constant.
For the case when both $k$ and $d$ are constant, Inaba et al. \cite{Inaba1994} showed that \kmeans can be solved in polynomial 
time. 
For fixed $k$ (but arbitrary $d$), several PTASs have been proposed, each with some improvement over past results in terms of 
running time; e.g., see \cite{DKKR03,FMS07,HK05,HM04,KSS04,KSS10}.
Despite a large number of PTASs for \kmeans with fixed $k$, obtaining a PTAS for \kmeans in fixed dimensional Euclidean space
 had been an open problem for a along time.
Bandyapadhyay and Varadarajan~\cite{BV16} presented a bicriteria PTAS for the problem that finds a $(1+\epsilon)$-approximation
 solution which might use up to
$(1+\epsilon)k$ clusters. The first true PTAS for the problem was recently obtained by \cite{AKM16B,FRS16B} via local search. 
The authors show that their analysis also works for metrics with fixed
 doubling dimension \cite{FRS16B}, and the shortest path metrics of minor closed graphs \cite{AKM16B}.
 
There are several constant factor approximation algorithms for \kmed in general metrics. 
The simple local search (identical with the one for \kmeans) is known to give a $3+\epsilon$ approximation by
 Arya et al.~\cite{AGKMMP01,AGKMMP04}. 
The current best approximation uses different techniques and has an approximation ratio
of $2.675+\epsilon$~\cite{LS13,BPRS14}. For Euclidean metrics, this was recently improved to $2.633+\epsilon$ via primal-dual
 algorithms \cite{ANSW16}.
Arora et al.~\cite{ARR98}, based on Arora's quadtree dissection~\cite{Arora98}, gave the first PTAS for \kmed in fixed dimensional 
Euclidean metrics. We note \cite{ARR98} also gives a PTAS for \uflo and \kmedo in constant-dimensional Euclidean metrics, our results
for Euclidean metrics in particular are therefore most meaningful for \kmeanso.
The recent PTASs (based on local search) for \kmed  by \cite{AKM16B,FRS16B} work also for doubling metrics
\cite{FRS16B} and also for minor-closed metrics \cite{AKM16B}.
No PTAS or bicriteria PTAS was known for such metrics for even uniform-cost UFL with outliers (\uuflo) or \kmedo.

Currently, the best approximation for \ufl in general metrics is a 1.488-approximation~\cite{L11}. As with \kmed, PTASs are known 
 for \ufl in fixed dimensional Euclidean metrics \cite{ARR98}, 
metrics with fixed doubling dimension \cite{FRS16B}, and for the shortest path metrics of minor closed graphs \cite{AKM16B}; 
however the results by \cite{AKM16B} only work for \ufl with uniform opening cost.

\subsection{Our results} 
We present a general method for converting local search analysis for clustering problems without outliers to problems with outliers.
Roughly speaking, we preprocess and then aggregate test swaps used in the analysis of such problems in order to incorporate outliers.
We demonstrate this by applying our ideas to \uuflo, \kmedo, and \kmeanso.

Most of our results are for metrics of fixed doubling dimensions as well as shortest path 
metrics of minor-closed graphs.
First, we show that on such metrics a simple multi-swap local search heuristic yields a PTAS for \uuflo.

\begin{theorem}\label{theo:uflo}
A $\rho=\rho(\epsilon,d)$-swap local search algorithm yields a PTAS for \uuflo for doubling metrics and minor-closed graphs.
Here, $d$ is either the doubling constant of the metric or a constant that depends on a minor that is exclude from the minor-closed family.
\end{theorem}

We then extend this result to \kmed and \kmeans with outliers (\kmedo and \kmeanso) and obtain bicriteria PTASs for them.
More specifically:

\begin{theorem}\label{theo:kmeanso}
A $\rho=\rho(\epsilon,d)$-swap local search algorithm yields a bicriteria PTAS for \kmedo and \kmeanso on doubling metrics
and minor-closed graphs; i.e. finds a solutions of
cost at most $(1+\epsilon)\cdot OPT$ with at most $(1+\epsilon)k$ clusters that uses at most $z$ outliers where $OPT$ is the
cost of optimum $k$-clustering with $z$ outliers.
\end{theorem}
In fact, in minor-closed metrics a true local optimum in the local search algorithm would find a solution using $(1+\epsilon)k$ clusters with cost at most $OPT$
in both \kmedo and \kmeanso, but a $(1+\epsilon)$-factor must be lost due to a standard procedure to ensure the local search algorithm terminates in polynomial time.

We show how these results can be extended to the setting where the metric is the $\ell^q_q$-norm, i.e. cost of connecting two
points $i,j$ is $\delta^q(i,j)$ (e.g. \kmedo is when $q=1$, \kmeanso is when $q=2$).
Finally, we show that in general metrics we still recover a bicriteria constant-factor approximation for \kmedo and \kmeanso. While not our main result,
it gives much more reasonable constants under bicriteria approximations for these problems.
\begin{theorem}\label{theo:general}
A $1/\epsilon^{O(1)}$-swap local search algorithm finds a solution using $(1+\epsilon)k$ clusters and has cost at most $(3+\epsilon) \cdot OPT$
for \kmedo or cost at most $(25+\epsilon) \cdot OPT$ for \kmeanso.
\end{theorem}
It should be noted that a true constant-factor approximation for \kmedo is given by Chen \cite{Chen07}, though the constant seems to be very large.
More interestingly, even a constant-factor bicriteria for \kmeanso that uses $(1+\epsilon)k$ clusters and discards the correct number of outliers has not been observed before (recall that the bicriteria algorithm of \cite{CKMN01} for \kmedo
has ratio $O(1/\epsilon)$ for \kmed using at most $(1+\epsilon)z$ outliers).
It is not clear that Chen's algorithm can be extended to give a true constant-factor approximation for \kmedo; one technical challenge is that part of the algorithm reassigns points multiple times
over a series of $O(\log n)$ iterations. So it is not clear that the algorithm in \cite{Chen07} extends to \kmeans.

To complement these results we show that for \uflo (i.e. non-uniform opening costs) any multi-swap local search
has unbounded gap. Also, for \kmedo and \kmeanso we show that 
without violating the number of clusters or outliers, any  multi-swap local search
will have unbounded gap even on Euclidean metrics.

\begin{theorem}\label{theo:localitygap}
Multi-swap local search has unbounded gap for \uflo and for \kmedo and \kmeanso on Euclidean metrics.
\end{theorem}

\noindent
{\bf Outline of the paper:}
We start with preliminaries and notation. Then in Section \ref{sec:uflo}, we prove Theorem \ref{theo:uflo}.
In Section \ref{sec:kmeanso} we prove Theorem \ref{theo:kmeanso} for the case of \kmed on doubling metrics
and then in \ref{sec:extensions} we show how
to extend these theorems to $\ell^q_q$-norm distances as well as minor-closed families of graphs. 
Theorem \ref{theo:general} is proven in Section \ref{sec:general}.
Finally, the proof of Theorem \ref{theo:localitygap} comes in Section \ref{sec:gap}.

\subsection{Preliminaries and Notation}\label{sec:prem}
In \uuflo we are given a set of $\cX$ points, a set $\fa$ of centres and $z$ for the number of outliers. Our goal is to select a set $C\subset\fa$ to open and a set $Z\subset\cX$ to
discard and assign each $j\in\cX-Z$ to the nearest centre in $C$ to minimize 
 $\sum_{x\in \cX \setminus Z}\delta(x,C) + |C|$, where $\delta(x,C)$ is the distance of $x$ to nearest
 $c\in C$.
In \kmedo and (discrete) \kmeanso, along with $\cX$, $\fa$, and $z$, we have an integer $k$ as the number of clusters. We like 
to find $k$ centres $C=\{c_1,\ldots,c_k \} \subseteq \fa$ and a set of (up to) $z$ points $Z$ 
to act as outliers. In \kmedo, we like to minimize  $\sum_{x\in \cX \setminus Z}\delta(x,C)$ and in \kmeanso, we want
to minimize $\sum_{x\in \cX \setminus Z}\delta(x,C)^2$.

For all these three problems, if we have the $\ell^q_q$-norm then we like to minimize
 $\sum_{x\in \cX \setminus Z}\delta^q(x,C)$.
 We should note that in classical \kmeans (in $\RR^d$), one is not given a candidate set of potential
centres, but they can be chosen anywhere. However,  
by using the classical result of~\cite{Matousek00}, at a loss of $(1+\epsilon)$ factor
we can assume we have a set $\fa$ of ``candidate'' centres from which the centres can be chosen
 (i.e. reduce to the discrete case considered here).
This set can be computed in time $O(n\epsilon^{-d}\log(1/\epsilon))$ and $|\fa|=O(n\epsilon^{-d}\log(1/\epsilon))$.

\section{Uniform-Cost UFL with Outliers in Doubling Metrics}\label{sec:uflo}
We start with presenting an approximation scheme for \uuflo in doubling metrics (Theorem \ref{theo:uflo}).
Recall there is already a PTAS for \uflo in constant-dimensional Euclidean metrics using dynamic programming for \ufl through quadtree decompositions \cite{ARR98}.
However, our approach generalizes to many settings where quadtree decompositions
are not known to succeed such as when the assignment cost between a point $j$ and a centre $i$ is $\delta(j,i)^q$ for constant 
$1 < q < \infty$
including \kmeans distances ($q = 2$) and also to shortest-path metrics of edge-weighted minor-closed graphs.
Furthermore, our technique here extends to \kmeanso (as seen in the next section).
Still, we will initially present our approximation scheme in this simpler setting to lay the groundwork and introduce the main ideas.

Recall that we are given a set $\cX$ of points and a set $\fa$ of possible centres in a metric space with doubling dimension $d$ and a number $z$ bounding the number of admissible outliers.
As the opening costs are uniform, we may scale all distances and opening costs so the opening cost of a centre is 1.
For any $\emptyset \subsetneq \loc \subseteq \fa$, order the points $j \in \cX$ as $j^\loc_1, j^\loc_2, \ldots, j^\loc_n$ in nondecreasing order
of distance $\delta(j^\loc_i, \loc)$.
The cost of $\loc$ is then $\cost(\loc) := \sum_{\ell=1}^{n-z} \delta(j^\loc_\ell, \loc) + |\loc|$.
That is, after discarding the $z$ points that are furthest from $\loc$ the others are assigned to the nearest centre in 
$\loc$: we pay this total assignment cost for all points that are not outliers and also the total centre opening cost $|\loc|$.
The goal is to find $\emptyset \subsetneq \loc \subseteq \fa$ minimizing $\cost(\loc)$.

Let $\epsilon > 0$ be a constant. Let $\rho' := \rho'(\epsilon, d)$ be some constant we will specify later.
We consider a natural multiswap heuristic for \uuflo, described in Algorithm \ref{alg:local}.

\begin{algorithm*}[t]
 \caption{\uufl $\rho'$-Swap Local Search} \label{alg:local}
\begin{algorithmic}
\State Let $\loc$ be an arbitrary non-empty subset of $\fa$
\While{$\exists$ sets $P\subseteq\fa-\loc$, $Q\subseteq \loc$ with $|P|,|Q|\leq \rho'$ s.t. 
$\cost((\loc - Q) \cup P)<\cost(\loc)$}
\State $\loc\leftarrow (\loc - Q) \cup P$
\EndWhile
\State \Return $\loc$
\end{algorithmic}
\end{algorithm*}

Each iteration can be executed in time $|\cX| \cdot |\fa|^{O(\rho')}$. It is not clear that the number of iterations is bounded by a polynomial.
However, the standard trick from \cite{AGKMMP01,KMNPSW04} works in our setting. That is, in the loop condition we instead
perform the swap only if $\cost((\loc - Q)\cup P) \leq (1+\epsilon/|\fa|) \cdot \cost(\loc)$. This ensures the running time is polynomial
in the input size as every $|\fa|/\epsilon$ iterations the cost decreases by a constant factor.

Our analysis of the local optimum
follows the standard template of using test swaps to generate inequalities to bound $\cost(S)$. The total number of swaps we use
to generate the final bound is at most $|\fa|$, so (as in \cite{AGKMMP01,KMNPSW04}) the approximation guarantee of a local optimum will only be degraded by an additional $(1+\epsilon)$-factor.
For the sake of simplicity in our presentation we will bound the cost of a local optimum solution returned by Algorithm \ref{alg:local}.


\subsection{Notation and Supporting Results from Previous Work}\label{sec:recall}
We use many results from \cite{FRS16B}, so we use the same notation. In this section, these results are recalled and a quick 
overview of the analysis of the multiswap local search heuristic for \uufl
is provided; this is simply Algorithm \ref{alg:local} where the cost function is defined appropriately for \uufl. Recall that
 in \ufl, each centre $i \in \fa$ has an {\em opening cost} $f_i \geq 0$.
Also let $\cost(\loc) = \sum_{j \in \cX} \delta(j, \loc) + \sum_{i \in \loc} f_i$. In \uufl all opening costs are uniform.

Let $\loc$ be a local optimum solution returned by Algorithm \ref{alg:local} and $\opt$ be a global optimum solution.
As it is standard in local search algorithms for uncapacitated clustering problems, we may assume $\loc \cap \opt = \emptyset$. 
This can be assumed by duplicating each $i \in \fa$, asserting $\loc$ uses
only the original centres and $\opt$ uses only the copies. It is easy to see $\loc$ would still be a local optimum solution.
Let $\sigma : \cX \rightarrow \loc$ be the point assignment in the local optimum and $\sigma^* : \cX \rightarrow \opt$ be the 
point assignment in the global optimum.
For $j \in \cX$, let $c_j = \delta(j, \loc)$ and $c^*_j = \delta(j, \opt)$ (remember there are no outliers in this review of 
\cite{FRS16B}).

For each $i \in \loc \cup \opt$, let $D_i$ be the distance from $i$ to the nearest centre of the other {\em type}. That is, for 
$i \in \loc$ let $D_i = \delta(i, \opt)$ and for $i^* \in \opt$ let $D_{i^*} = \delta(i^*, \loc)$.
Note for every $j \in \cX$ and every $i' \in \{\sigma(j), \sigma^*(j)\}$ that
\begin{equation} \label{eqn:radius}
D_{i'} \leq \delta(\sigma(j), \sigma^*(j)) \leq c^*_j + c_j.
\end{equation}

A special pairing $\cents \subseteq \loc \times \opt$ was identified in \cite{FRS16B} (i.e. each $i \in \loc \cup \opt$ appears at most once among pairs in $\cents$)
with the following special property.
\begin{lemma}[Lemma 3 in \cite{FRS16B}, paraphrased]\label{lem:cents}
For any $A \subseteq \loc \cup \opt$ such that $A$ contains at least one centre from every pair in $\cents$, $\delta(i, A) \leq 5 \cdot D_i$ for every $i \in \loc \cup \opt$.
\end{lemma}

Next, a {\em net} is cast around each $i \in \loc$. The idea is that any swap that has $i \in \loc$ being closed would have something open near every $i^* \in \opt$ that itself is close to $i$.
More precisely, \cite{FRS16B} identifies a set $\net \subseteq \loc \cup \opt$ with the following properties. For each $i \in \loc$ and $i^* \in \opt$ with $\delta(i, i^*) \leq D_i/\epsilon$ and $D_{i^*} \geq \epsilon \cdot D_i$
there is some pair $(i, i') \in \net$ with $\delta(i', i^*) \leq \epsilon \cdot D_{i^*}$. The set $\net$ contains further properties to enable Theorem \ref{thm:partition} (below), but these  are sufficient for our discussion.

The last major step in \cite{FRS16B} before the final analysis was to provide a structure theorem showing $\loc \cup \opt$ can be partitioned into test swaps that mostly allow the redirections discussed above.
\begin{theorem}[Theorem 4 in \cite{FRS16B}, slightly adjusted]\label{thm:partition}
For any $\epsilon > 0$, there is a constant $\rho := \rho(\epsilon, d)$ and
a randomized algorithm that samples a partitioning $\pi$ of $\opt \cup \loc$ such that:
\begin{itemize}
\item For each part $P \in \pi$, $|P \cap \opt|, |P \cap \loc| \leq \rho$.
\item For each part $P \in \pi$, $\loc \triangle P$ includes at least one centre from every pair in $\cents$.
\item For each $(i^*, i) \in \net$, $\Pr[i, i^* {\rm ~lie~in~different~parts~of~}\pi] \leq \eps$.
\end{itemize}
\end{theorem}
There is only a slight difference between this statement and the original statement in \cite{FRS16B}. Namely, the first condition of Theorem 4 in \cite{FRS16B} also asserted $|P \cap \opt| = |P \cap \loc|$.
As noted at the end of \cite{FRS16B}, this part can be dropped by skipping one final step of the proof that ``balanced'' parts of the partition that were constructed (also see \cite{FRS16} for further details).

The analysis in \cite{FRS16B} used Theorem \ref{thm:partition}
to show $\cost(\loc) \leq (1+O(\epsilon)) \cdot OPT$ generated an inequality by swapping each part $P \in \pi$. Roughly speaking, for each $j \in \cX$ with probability at least $1-\epsilon$
(over the random construction of $\pi$), the swap that closes $\sigma(j)$ will open something very close to $\sigma(j)$ or very close to $\sigma^*(j)$. With the remaining probability, we can at least move $j$
a distance of at most $O(c^*_j + c_j)$. Finally, if $j$ was never moved to something that was close to $\sigma^*(j)$ this way then we ensure we move $j$ from $\sigma(j)$ to $\sigma^*(j)$ when $\sigma^*(j)$ is swapped in.

In our analysis for clustering with outliers, our reassignments for points that are not outliers in either the local or global optimum are motivated by this approach. Details will appear
below in our analysis of Algorithm \ref{alg:local}.


\subsection{Analysis for Uniform-Cost UFL with Outliers: An Outline}
Now let $\loc$ be a locally optimum solution for Algorithm \ref{alg:local}, let $\cX^a$ be the points in $\cX$ that are assigned to $\loc$ and $\cX^o$ be the points in $\cX$ that are outliers when opening $\loc$.
Similarly, let $\opt$ be a globally optimum solution, let $\cX^{a^*}$ be the points in $\cX$ that are assigned to $\opt$ and $\cX^{o^*}$ be the points in $\cX$ that are outliers when opening $\opt$.
Note $|\cX^o| = |\cX^{o^*}| = z$.

Let $\sigma : \cX \rightarrow \loc \cup \{\bot\}$ assign $j \in \cX^a$ to the nearest centre in $\loc$ and $j \in \cX^o$ to $\bot$. Similarly, let $\sigma^* : \cX \rightarrow \opt \cup \{\bot\}$ map each $j \in \cX^{a^*}$ to the nearest centre in $\opt$
and each $j \in \cX^{o^*}$ to $\bot$. For $j \in \cX$, we let $c_j = 0$ if $j \in \cX^o$ and, otherwise, let $c_j = \delta(j, \loc) = \delta(j, \sigma(j))$. Similarly, let $c^*_j = 0$ if $j \in \cX^{o^*}$ and, otherwise,
let $c^*_j = \delta(j, \opt) = \delta(j, \sigma(j))$.

Our starting point is the partitioning scheme described in Theorem \ref{thm:partition}. The new issue to be handled is in reassigning the outliers when a part is swapped. That is,
for any $j \in \cX^{o^*}$ any swap that has $\sigma(j)$ swapped out cannot, in general, be reassigned anywhere cheaply.

Really the only thing we can do to upper bound the assignment cost change for $j$ is to make it an outlier.
We can try assigning each $j \in \cX^{o}$ to $\sigma^*(j)$ if it is opened, thereby allowing one $j \in \cX^{o^*}$ with $\sigma(j)$ being swapped out to become an outlier.
However there may not be enough $j' \in \cX^{o}$ that have $\sigma^*(j)$ opened after the swap. That is, we might not 
remove enough outliers from the solution $\loc$ to be able to let all
such $j$ become outliers.

Our approach is to further combine parts $P$ of the partition $\pi$ and perform the swaps simultaneously for many of these parts. Two of these parts in a larger grouping will not actually be swapped out:
their centres in $\opt$ will be opened to free up more spaces for outliers yet their centres in $\loc$ will not be swapped out. Doing this carefully, we ensure that the total number of $j \in \cX^{o^*}$ that have $\sigma(j)$ being closed
is at most the total number of $j \in \cX^{o}$ that have $\sigma^*(j)$ being opened.

These larger groups that are obtained by combing parts of $\pi$ are not disjoint. However, the overlap of centres between larger groups will be negligible compared to $|\loc| + |\opt|$.


\subsection{Grouping the Parts}
We will assume $\cX^{o} \cap \cX^{o^*} = \emptyset$. This is without loss of generality as $\loc$ would still be a local optimum in the instance with $\cX^{o} \cap \cX^{o^*}$ removed and $z$ adjusted.
Recall we are also assuming $\loc \cap \opt = \emptyset$.

For each part $P$ of $\pi$, let $\Delta_P := |\{j \in \cX^{o} : \sigma^*(j) \in P\}| - |\{j \in \cX^{o^*} : \sigma(j) \in P\}|$. This is the difference between the number of outliers we can reclaim by moving them to $\sigma^*(j)$
(if it is open after swapping $P$) and the number of outliers $j$ that we must create because $\sigma(j)$ was closed when swapping $P$.

Consider the following refinements of $\pi$: $\pi^+ = \{P \in \pi : \Delta_P > 0\}, \pi^- = \{P \in \pi : \Delta_P < 0\}$ and $\pi^0 = \{P \in \pi : \Delta_P = 0\}$. Intuitively, nothing more needs to be done to prepare parts $P \in \pi^0$ for swapping
as this would create as many outliers as it would reclaim in our analysis framework. We work toward handling $\pi^+$ and $\pi^-$.

Next we construct a bijection $\kappa : \cX^{o} \rightarrow \cX^{o^*}$. We will ensure when $\sigma(j)$ is swapped out for some $j \in \cX^{o^*}$ that $\sigma^*(\kappa^{-1}(j))$ will be swapped in. So there is space to make $j$ an
outlier in the analysis. There are some cases in our analysis
where we never swap out $\sigma(j)$ for some $j \in \cX^{o^*}$, but we will still ensure $\sigma^*(\kappa^{-1}(j))$ is swapped in at some point so we can still make $j$ an outlier while removing $\kappa^{-1}(j)$ as an outlier
to get the negative dependence on $c_j$ in the final inequality.

To start defining $\kappa$, for each $P \in \pi$ we pair up points
in $\{j \in \cX^{o^*} : \sigma(j) \in P\}$ and $\{j \in \cX^o : \sigma^*(j) \in P\}$ arbitrarily until one of these two groups is exhausted. These pairs define some mapping of $\kappa$.
The number of unpaired points in $\{j \in \cX^{o^*} : \sigma(j) \in P\}$
is exactly $-\Delta_P$ if $\Delta_P < 0$ and the number of unpaired points in $\{ j \in \cX^o : \sigma^*(j) \in P\}$ is exactly $\Delta_P$.

Having done this for each $P$, we begin pairing unpaired points in $\cX^{o} \cup \cX^{o^*}$ between 
parts.
Arbitrarily order $\pi^+$ as $P_1^+, P_2^+, \ldots, P_m^+$ and $\pi^-$ as $P_1^-, P_2^-, \ldots, P_\ell^-$.
We will complete the pairing $\kappa$ and also construct edges in a bipartite graph $H$ with $\pi^+$ on one side and $\pi^-$ on the other side using Algorithm \ref{alg:pairing}.
To avoid confusion with edges in the distance metric, we call the edges of $H$ between $\pi^+$ and $\pi^-$ {\em superedges}.

In Algorithm \ref{alg:pairing}, we say a part $P \in \pi^+ \cup \pi^-$ has an unpaired point $j \in \cX^{o^*} \cup \cX^{o}$ and that $j$ is an unpaired point of $P$
if, currently, $\kappa$ has not paired $j$ and $\sigma(j) \in P$ or $\sigma^*(j) \in P$ (whatever is relevant).
The resulting graph over $\pi^+, \pi^-$ is depicted in Figure \ref{fig:grouping}, along with other features described below.

\begin{algorithm*}[t]
 \caption{Pairing Unpaired Outliers} \label{alg:pairing}
\begin{algorithmic}
\State $\cP \leftarrow \emptyset$ \Comment{A set of superedges between $\pi^+$ and $\pi^-$.}
\State $a \leftarrow 1, b \leftarrow 1$
\While{there are unpaired points}
\State Arbitrarily pair up (via $\kappa$) unpaired points between $P^+_a$ and $P^-_b$ until no longer possible.
\State $\cP \leftarrow \cP \cup \{(P^+_a, P^-_b)\}.$
\If{$P^+_a$ has no unpaired point}
\State $a \leftarrow a+1$
\EndIf
\If{$P^-_b$ has no unpaired point}
\State $b \leftarrow b+1$
\EndIf
\EndWhile
\State \Return $\loc$
\end{algorithmic}
\end{algorithm*}

\begin{figure}[th]
\begin{center}
\includegraphics[scale=0.53]{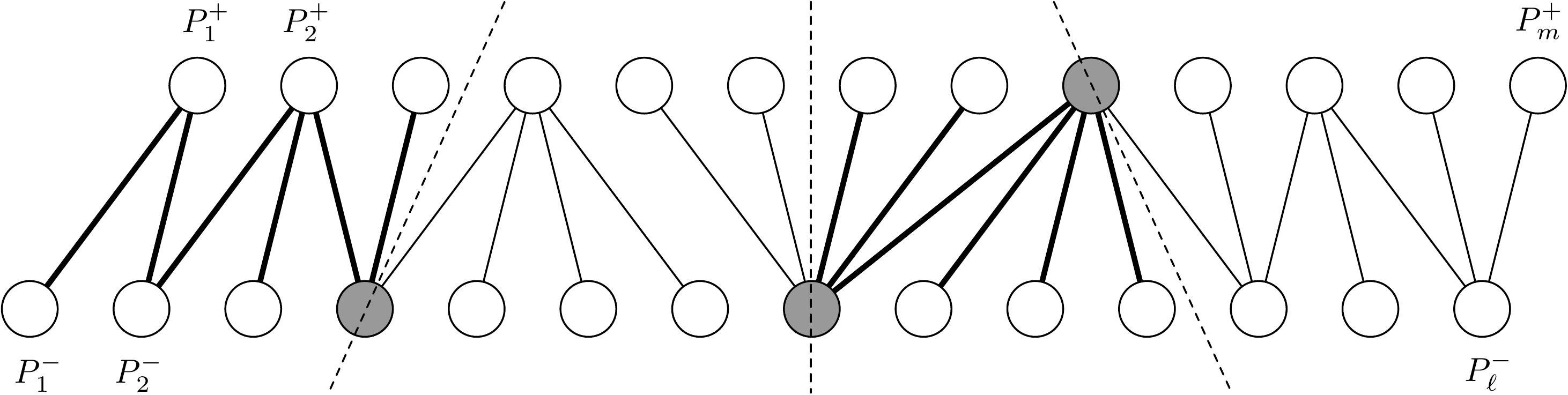}
\end{center}
\caption{The bipartite graph with sides $\pi^+, \pi^-$ and edges $\cP$. The edges are ordered left-to-right in order of their creation time, so $e_0$ is the leftmost edge.
A grouping with group size $\alpha = 6$ is depicted. The edges in $E_1$ and $E_3$ are bold.
Note the last group has more than $\alpha$ edges.
The parts that are split by some group are shaded. While not depicted, a part could be split by many groups (if it has very high degree in the bipartite graph).}\label{fig:grouping}
\end{figure}

We now group some superedges together.
Let $\alpha = 4\rho/\epsilon$ denote the {\em group size}. 
We will assume $|\cP| > \alpha$, as otherwise $|\pi^+ \cup \pi^-| \leq 2\alpha$ and we could simply merge all parts in $\pi^+ \cup \pi^-$ into a single part $P'$ with $\Delta_{P'} = 0$ with $|P' \cap \loc|, |P' \cap \opt| \leq 2\alpha\rho$.
The final local search algorithm will use swap sizes greater than $2\alpha\rho$ and the analysis showing $\cost(\loc) \leq (1+O(\epsilon)) \cdot \cost(\opt)$ would then follow almost exactly as we show\footnote{One very
minor modification is that the presented proof of Lemma \ref{lem:facs} in the final analysis ultimately uses $|\cP| \geq \alpha$. It still holds otherwise in a trivial way as there would be no overlap between groups.}.

Order edges $e_0, e_1, \ldots, e_{|\cP|-1}$ of $\cP$ according to when they were formed.
For each integer $s \geq 0$ let $E_s = \{e_i : \alpha \cdot s \leq i < \alpha \cdot (s+1)\}$. Let $s'$ be the largest index with $E_{s'+1} \neq \emptyset$ (which exists by the assumption $|\cP| > \alpha$). Merge the last two groups by replacing
$E_{s'}$ with $E_{s'} \cup E_{s'+1}$. Finally, for each $0 \leq s \leq s'$ let $G_s \subseteq \loc \cup \opt$ consist of all centres $i \in \loc \cup \opt$ belonging to a part $P$ that is an endpoint of some superedge in $E_s$.
The grouping is $\cG = \{G_s : 0 \leq s \leq s'\}$.
The groups of superedges $E_s$ are depicted in Figure \ref{fig:grouping}.
Note each part $P \in \pi^+ \cup \pi^-$ is contained in at least one group of $\cG$ because $\Delta_P \neq 0$ (so a
superedge edge in $\cP$ was created with $P$ as an endpoint). The following lemma is easy (see Appendix \ref{app:missingProofs}).

\begin{lemma}\label{lem:size}
For each $G_s \in \cG$, $\alpha-1 \leq |G_s| \leq 8 \rho\alpha$.
\end{lemma}

{\bf Note:} In fact, it is not hard to see that $H$ is a forest where each component of it is a caterpillar (a tree in which
every vertex is either a leaf node or is adjacent to a ``stalk'' node; stalk nodes form a path). The parts $P$ that are split
between different groups $G_s$ are the high degree nodes in $H$ and these parts belong to the ``stalk'' of $H$.

\begin{definition}
For each group $G_s \in \cG$, say a part $P\in\pi$ is split by $G_s$ if $P \subseteq G_s$ and 
$N_{\cP}(P) \not\subseteq E_s$, where $N_{\cP}(P)$ are those parts that have a superedge to $P$.
\end{definition}
We simply say $P$ is split if the group $G_s$ is clear from the context. The following lemma is easy (see Appendix \ref{app:missingProofs}
for proof).

\begin{lemma}\label{lem:split}
For each $G_s \in \cG$, there are at most two parts $P$ split by $G_s$.
\end{lemma}

\subsection{Analyzing Algorithm \ref{alg:local}}
Now suppose we run Algorithm \ref{alg:local} using $\rho' = 4\alpha\rho$. For each $P \in \pi^0$, extend $\cG$ to include a ``simple'' group $G_s = P$ for each $P \in \pi^0$ where $s$ is the next unused index.
Any $P \in \pi^0$ is not split by any group.
Each $P \in \pi$ is then contained in at least one group of $\cG$ and $|G_s \cap \loc|, |G_s \cap \opt| \leq \rho'$ for each $G_s \in \mathcal G$ by Lemma \ref{lem:size}.

We are now ready to describe the swaps used in our analysis of Algorithm \ref{alg:local}. Simply, for $G_s \in \cG$ let $\loc_s$ be the centres in $G_s \cap \loc$ that are not in a part $P$
that is split and let $\opt_s$ simply be $G_s \cap \opt$. We consider the swap $\loc \rightarrow (\loc - \loc_s) \cup \opt_s$.

To analyze these swaps, we further classify each $j \in \cX^a \cap \cX^{a*}$ in one of four ways. This is the same classification from \cite{FRS16B}.
Label $j$ according to the first property that it satisfies below.
\begin{itemize}
\item {\bf lucky}: both $\sigma(j)$ and $\sigma^*(j)$ in the same part $P$ of $\pi$.
\item {\bf long}: $\delta(\sigma(j), \sigma^*(j)) > D_{\sigma(j)}/\epsilon$.
\item {\bf good}: either $D_{\sigma^*(j)} \leq \epsilon D_{\sigma(j)}$ or there is some $i' \in \opt$ with $\delta(\sigma^*(j), i') \leq \epsilon \cdot D_{\sigma^*(j)}$ and $(\sigma(j), i') \in \net$ where both $\sigma(j)$ and $i'$ lie in the same part.
\item {\bf bad}: $j$ is not lucky,  long, or good. Note, by Theorem \ref{thm:partition}, that 
$\Pr[j\;\text{ is bad}] \leq \epsilon$ over the random construction of $\pi$.
\end{itemize}
Finally, as a technicality for each $j \in \cX^a \cap \cX^{a^*}$ where $\sigma^*(j)$ lies in
a part that is split by some group and $j$ is either lucky or long, let $s(j)$ be any index such that group $G_{s(j)} \in \cG$
contains the part with $\sigma^*(j)$. The idea is that we will reassign $j$ to $\sigma^*(j)$ only when group $G_{s(j)}$ is processed.

Similarly, for any $j \in \cX^o$, let $s(j)$ be any index such that $\sigma^*(j), \sigma(\kappa(j)) \in G_{s(j)}$. The following lemma shows this is always possible.
\begin{lemma}\label{lem:outliers}
For each $j \in \cX^o$ there is at least one group $G_s \in \cG$ where $\sigma^*(j), \sigma(\kappa(j)) \in G_s$.
\end{lemma}
\begin{proof}
If $\sigma^*(j)$ and $\sigma(\kappa(j))$ lie in the same part $P$, this holds because each part is a subset of some group. Otherwise, $j$ is paired with $\kappa(j)$ at some point in Algorithm \ref{alg:pairing} and an edge $(P^+_a, P^-_b)$
is added to $\cP$ where $\sigma^*(j) \in P^+_a, \sigma(\kappa(j)) \in P^-_b$. The centres in both endpoints of this 
super edge were added to some group $G_s$.
\end{proof}

We now place a bound on the cost change in each swap. Recall $|G_s \cap \loc|, |G_s \cap \opt| \leq \rho'$ and $\loc$ is a local optimum, so
\begin{equation*}\label{eqn:lower}
0 \leq \cost((\loc \triangle \loc_s) \cup \opt_s) - \cost(\loc).
\end{equation*}
We describe a feasible reassignment of points to upper bound the cost change. This may cause some points in $\cX^a$ becoming outliers and other points in $\cX^o$ now being assigned.
We take care to ensure the number of points that are outliers in our reassignment is exactly $z$, as required.

Consider the following instructions describing one possible way to reassign a point $j \in \cX$ when processing $G_s$. This may not describe an optimal reassignment, but it places an upper bound on the cost change.
First we describe which points should be moved to $\sigma^*(j)$ if it becomes open.
Note that the points in $\cX^o\cup \cX^{o^*}$ are paired via $\kappa$ (and $\cX^o\cap \cX^{o^*}=\emptyset$).
Below we specify what to do for each point $j\in\cX^o$ and $\kappa(j)$ together.
\begin{itemize}
\item If $j \in \cX^o$ and $s = s(j)$, then make $\kappa(j)$ an outlier and connect $j$ to $\sigma^*(j)$, which is now open. The total assignment cost change for $j$ and $\kappa(j)$ is $c^*_j - c_{\kappa(j)}$.
\end{itemize}

Note that so far this reassignment still uses $z$ outliers because each new outlier $j'$ has its paired point $\kappa^{-1}(j')$ that used to be an outlier become connected. The rest of the analysis will not create any more outliers.
The rest of the cases are for when $j\in\cX^a\cap\cX^{a^*}$.

\begin{itemize}
\item If $j$ is lucky or long and $s = s(j)$, reassign $j$ from $\sigma(j)$ to $\sigma^*(j)$. The 
assignment cost change for $j$ is $c^*_j - c_j$.

\item If $j$ is long, move $j$ to the open centre that is nearest to $\sigma(j)$. By Lemma \ref{lem:cents} and because $j$ is long, the assignment cost increase for $j$ can be bounded as follows:
\[ 5 \cdot D_{\sigma(j)} \leq 5 \epsilon \cdot \delta(\sigma(j), \sigma^*(j)) \leq 5\epsilon \cdot (c^*_j + c_j). \]

\item If $j$ is good and $D_{\sigma^*(j)} \leq \epsilon \cdot D_{\sigma(j)}$, move $j$ to the open centre that is nearest to $\sigma^*(j)$. By \eqref{eqn:radius} in Section \ref{sec:recall} and Lemma \ref{lem:cents},
the assignment cost increase for $j$ can be bounded as follows:
\[c^*_j + 5 \cdot D_{\sigma^*(j)} - c_j \leq c^*_j + 5 \epsilon \cdot D_{\sigma(j)} - c_j 
\leq c^*_j + 5\epsilon\cdot (c^*_j + c_j)  - c_j 
 = (1+5\epsilon) \cdot c^*_j - (1-5\epsilon)\cdot c_j.\]

\item If $j$ is good but $D_{\sigma^*(j)} > \epsilon \cdot D_{\sigma(j)}$, then let $i'$ be such that $\sigma(j),i'$ both lie in $G_s$ and $\delta(\sigma^*(j), i') \leq \epsilon \cdot D_{\sigma^*(j)}$. Reassigning $j$ from to $i'$
bounds its assignment cost change by
\[ c^*_j + \delta(\sigma^*(j), i')  -c_j \leq c^*_j + \epsilon \cdot D_{\sigma^*(j)} - c_j \leq (1+\epsilon) \cdot c^*_j - (1-\epsilon) \cdot c_j. \]

\item Finally, if $j$ is bad then simply reassign $j$ to the open centre that is nearest to $\sigma(j)$. By \eqref{eqn:radius} and Lemma \ref{lem:cents}, the assignment cost for $j$ increases by at most
$5 \cdot D_{\sigma(j)} \leq 5 \cdot (c^*_j + c_j).$ 

This looks large, but its overall contribution to the final analysis will be scaled by an $\epsilon$-factor because a point is bad only with probability at most $\epsilon$ over the random sampling of $\pi$.
\end{itemize}

Note this accounts for all points $j$ where $\sigma(j)$ is closed. Every other point $j$ may stay assigned to $\sigma(j)$ to bound its assignment cost change by 0.

For $j \in \cX$, let $\Delta_j$ denote the total reassignment cost change over all swaps when moving $j$ as described above. This should not be confused with the previously-used notation $\Delta_P$ for a part $P \in \pi$.
We bound $\Delta_j$ on a case-by-case basis below.
\begin{itemize}
\item If $j \in \cX^{o}$ then the only time $j$ is moved is for the swap involving $G_{s(j)}$. So $\Delta_j = c^*_j$.
\item If $j \in \cX^{o^*}$ then the only time $j$ is moved is for the swap involving $G_{s(\kappa^{-1}(j))}$. So $\Delta_j = -c_j$.
\item If $j$ is lucky then it is only moved when $G_{s(j)}$ is processed so $\Delta_j = c^*_j - c_j$.
\item If $j$ is long then it is moved to $\sigma^*(j)$ when $G_{s(j)}$ is processed and it is moved near $\sigma(j)$ when $\sigma(j)$ is closed, so
$\Delta_j \leq c^*_j - c_j + 5\epsilon \cdot (c^*_j + c_j) = (1+5\epsilon) \cdot c^*_j - (1-5\epsilon) \cdot c_j.$
\item If $j$ is good then it is only moved when $\sigma(j)$ is closed so $\Delta_j \leq (1+5\epsilon) \cdot c^*_j - (1-5\epsilon) \cdot c_j$.
\item If $j$ is bad then it is only moved when $\sigma(j)$ is closed so $\Delta_j \leq 5 \cdot (c^*_j + c_j)$.
\end{itemize}

To handle the centre opening cost change, we use the following fact.
\begin{lemma}\label{lem:facs}
$\sum_{G_s \in \cG} |\opt_s| - |\loc_s| \leq (1+2\epsilon) \cdot |\opt| - (1-2\epsilon) \cdot |\loc|$
\end{lemma}
\begin{proof}
For each $G_s \in \cG$, let $P_s$ be the union of all parts used to form $G_s$ that are not split by $G_s$.
Let $\overline P_s = G_s - P_s$, these are centres in $G_s$ that lie in a part split by $G_s$.

Now, $|\overline P_s| \leq 2\rho$ because at most two parts are split by $G_s$ by Lemma \ref{lem:split}.
On the other hand, by Lemmas \ref{lem:size} and \ref{lem:split} there are at least $\alpha-3$ parts that were used to form $G_s$ that were not split by $G_s$. As each part contains at least one centre, then $|P_s| \geq \alpha-3$.
Thus, for small enough $\epsilon$ we have
\[ |\overline P_s| \leq 2\rho \leq \epsilon \cdot (\alpha-3) \leq \epsilon |P_s|. \]

Note $\sum_{G_s \in \cG} |P_s| \leq |\loc| + |\opt|$ because no centre appears in more than one set of the form $P_s$. Also note $|\opt_s| \leq |P_s \cap \opt| + |\overline P_s|$ and $|\loc_s| \geq |G_s \cap \loc| - |\overline P_s|$,
\[ \sum_{G_s \in \cG} |\opt_s| - |\loc_s| \leq \sum_{G_s \in \cG} |P_s \cap \opt| - |G_s \cap \loc| + 2\epsilon |P_s| \leq |\opt| - |\loc| + 2\epsilon \cdot (|\opt| + |\loc|). \]

\end{proof}

Putting this all together,
\begin{alignat*}{2}
0 \quad & \leq \quad \sum_{G_s \in \cG} \cost((\loc - \loc_s) \cup \opt_s) - \cost(\loc) \\
& \leq \quad \sum_{j \in \cX} \Delta_j + \sum_{G_s \in \cG} |\opt_s| - |\loc_s| \\
& \leq \quad \sum_{\begin{array}{cc}
j \in \cX^{a} \cap \cX^{a^*} \\
j\text{ is not bad}
\end{array}} \left[(1+5\epsilon) \cdot c^*_j - (1-5\epsilon) c_j\right] +  \sum_{j\text{ bad}}5(c^*_j + c_j)+ \\
 &\qquad \sum_{j \in \cX^o} c^*_j - \sum_{j \in \cX^{o^*}} c_j + (1+2\epsilon) \cdot |\opt| - (1-2\epsilon) \cdot |\loc|.
\end{alignat*}

This holds for any $\pi$ supported by the partitioning scheme described in Theorem \ref{thm:partition}. Taking expectations over the random choice of $\pi$
and recalling $\Pr[j\text{is bad}] \leq \epsilon$ for any $j \in \cX^a \cap \cX^{a^*}$,
\[ 0 \leq \sum_{j \in \cX} \left[(1+10\epsilon) \cdot c^*_j - (1-10\epsilon) \cdot c_j\right] + (1+2\epsilon) \cdot |\opt| - (1-2\epsilon) \cdot |\loc|. \]
Rearranging and relaxing slightly further shows
\[ (1-10\epsilon) \cdot \cost(\loc) \leq (1+10\epsilon) \cdot \cost(\opt). \]
Ultimately, $\cost(\loc) \leq (1+30\epsilon) \cdot OPT.$


\section{$k$-Median and $k$-Means with Outliers}\label{sec:kmeanso}
In this section we show how the results of the previous section can be extended to get a bicriteria
approximation scheme for \kmedo and \kmeanso with outliers in doubling metrics (Theorem \ref{theo:kmeanso}). 
For ease of exposition we present the
result for \kmedo. More specifically, given a set $\cX$ of points, set $\fa$ of possible centers, positive
integers $k$ as the number of clusters and $z$ as the number of outliers for \kmed,
we show that a $\rho'$-swap local search (for some $\rho'=\rho(\epsilon,d)$
to be specified) returns a solution of cost at most $(1+\epsilon)OPT$ using at most $(1+\epsilon)k$ centres (clusters)
and has at most $z$ outliers.
Note that a local optimum $\loc$ satisfies $|\loc| = (1+\epsilon) \cdot k$ unless $\cost(\loc) = 0$ already, in which case our analysis is already done.
Extension of the result ot \kmeans or in general to a clustering where the
distance metric is the $\ell^q_q$-norm is fairly easy and is discussed in the next section where we also show how
we can prove the same result for shortest path metric of minor-closed families of graphs.

\begin{algorithm*}[t]
 \caption{\kmed $\rho'$-Swap Local Search} \label{alg2:local}
\begin{algorithmic}
\State Let $\loc$ be an arbitrary set of $(1+\epsilon)k$ centres from $\fa$
\While{$\exists$ sets $P\subseteq\fa-\loc$, $Q\subseteq \loc$ with $|P|,|Q|\leq \rho'$ s.t. 
$\cost((\loc - Q) \cup P)<\cost(\loc)$ and $|(\loc - Q)\cup P|\leq (1+\epsilon)k$}
\State $\loc\leftarrow (\loc - Q) \cup P$
\EndWhile
\State \Return $\loc$
\end{algorithmic}
\end{algorithm*}


The proof uses ideas from both \cite{FRS16B} for the PTAS for \kmeans as well as the results of the previous section
for \uuflo. Let $\loc$ be a local optimum solution returned by Algorithm \ref{alg2:local} 
that uses at most $(1+\epsilon)k$ clusters and has at most $z$ outliers and 
let $\opt$ be a global optimum solution to \kmedo with $k$ clusters and $z$ outliers.
Again, we use $\sigma:\cX\rightarrow \loc\cup\{\bot\}$ and $\sigma^*:\cX\rightarrow\opt\cup\{\bot\}$ as the assignment of points
to cluster centres, where $\sigma(j)=\bot$ (or $\sigma^*(j)=\bot$) indicates an outlier.
For $j\in\cX$ we use $c_j=\delta(j,\loc)$ and $c^*_j=\delta(j,\opt)$. The special pairing $\cents\subseteq \loc\times\opt$ and
Lemma \ref{lem:cents} as well as the notion of nets $\net\subseteq \loc\cup \opt$ are still used in our setting.
A key component of our proof is the following slightly modified version of Theorem 4 in \cite{FRS16B}:

\begin{theorem}\label{theo:newpartition}
For any $\epsilon >0$, there is a constant $\rho=\rho(\epsilon,d)$ and a randomized algorithm that samples a partitioning
$\pi$ of $\loc\cup\opt$ such that:
\begin{itemize}
\item For each part $P\in\pi$, $|P\cap \opt|<|P\cap\loc|\leq \rho$ 
\item For each part $P\in\pi$, $S\Delta P$ contains at least one centre from every pair in $\cents$
\item For each $(i^*,i)\in\net$,  $\Pr[i, i^* {\rm ~lie~in~different~parts~of~}\pi] \leq \eps$.
\end{itemize}
\end{theorem}

The only difference of this version and the one in \cite{FRS16B} is that in Theorem 4 in \cite{FRS16B}, for the first
condition we have $|P\cap\opt|=|P\cap\loc|\leq \rho$. The theorem was proved by showing a randomized partitioning that satisfies
conditions 2 and 3. To satisfy the 1st condition $|P\cap\opt|=|P\cap\loc|\leq \rho$ a balancing step was performed
at the end of the proof that would merge several (but constant) number of parts to get parts with equal number of 
centres from $\loc$ and $\opt$ satisfying the two other conditions.
Here, since $|\loc|=(1+\epsilon)k$ and $|\opt|=k$, we can modify the balancing step to 
make sure that each part has at least one more centre from $\loc$ than from $\opt$. We present a proof of this theorem
for the sake of completeness in Appendix \ref{app:newpartition}.

We define $\cX^a,\cX^o,\cX^{a^*},$ and $\cX^{o^*}$ as in the previoius section for \uuflo. 
Note that $|\cX^o|=|\cX^{o^*}|=z$.
As before, we assume $\cX^o\cap \cX^{o^*}=\emptyset$ and $\loc\cap\opt=\emptyset$. Let $\pi$ be a partitioning as in
Theorem \ref{theo:newpartition}. Recall that for each $P\in\pi$, 
$\Delta_P := |\{j \in \cX^{o} : \sigma^*(j) \in P\}| - |\{j \in \cX^{o^*} : \sigma(j) \in P\}|$; we define $\pi^+,\pi^-,\pi^0$
to be the parts with positive, negative, and zero $\Delta_P$ values. We define the bijection $\kappa:\cX^o\rightarrow\cX^{o^*}$
as before: for each $P\in\pi$, we pair up points (via $\kappa$)
in $\{j \in \cX^{o^*} : \sigma(j) \in P\}$ and $\{j \in \cX^o : \sigma^*(j)\} \in P$ arbitrarily.
Then (after this is done for each $P$) we begin pairing unpaired points in $\cX^{o} \cup \cX^{o^*}$ between groups using
Algorithm  \ref{alg:pairing}. The only slight change w.r.t. the previous section is in the grouping the super-edges of the
bipartite graph defined over $\pi^+\cup\pi^-$: 
we use parameter $\alpha=2\rho+3$ instead. Lemmas \ref{lem:size} and \ref{lem:split} still hold.

Suppose we run Algorithm \ref{alg2:local} with $\rho'=\alpha\rho$. As in the case of \uuflo, we add each $P\in\pi^0$ as a separate
group to $\cG$ and so each $P\in\pi$ is now contained in at least one group $G_s\in\cG$ and 
$|G_s\cap \loc|,|G_s\cap\opt|\leq\rho'$. For the points $j\in\cX^a\cap\cX^{a^*}$ (i.e. those that are not outliner in neither
$\loc$ nor $\opt$) the analysis is pretty much the same as the PTAS for standard \kmeans (without outliers) in \cite{FRS16B}.
We need extra care to handle the points $j$ that are outliers in one of $\loc$ or $\opt$ (recall that $\cX^o\cap\cX^{o^*}=\emptyset$.

As in the analysis of \uuflo, for $G_s\in\cG$, let $\loc_s$ be the centers in $G_s\cap\loc$ that 
are not in a part $P$ that is split, and $\opt_s=G_s\cap\opt$; consider the test swap that replaces $\loc$ with 
$(\loc-\loc_s)\cup\opt_s$. To see this is a valid swap considered by our algorithm,
note that the size of a test swap is at most $\rho'$. Furthremore,
there are at least $\alpha-3$ unsplit parts $P$ and at most two split parts in each $G_s$, and each unsplit part
has at least one more centre from $\loc$ than $\opt$ (condition 1 of Theorem \ref{theo:newpartition}); 
therefore, there are at least $\alpha-3$ more centres that are swapped out
in $\loc_s$ and these can account for the at most $2\rho$ centres of $\loc\cup\opt$ in the split parts that are not swapped out
(note that $\alpha-3=2\rho$).
Thus, the total number of centres of $\loc$ and $\opt$ after the test swap is still at most $(1+\epsilon)k$ and $k$, respectively.

We classify each $j\in\cX^a\cap\cX^{a^*}$ to {\bf lucky}, {\bf long}, {\bf good}, and {\bf bad} in
the same way as in the case of \uuflo. 
Furthermore, $s(j)$ is defined in the same manner: for each $j \in \cX^a \cap \cX^{a^*}$ where $\sigma^*(j)$ lies 
a part that is split by some group and $j$ is either lucky or long, let $s(j)$ be any index such that group $G_{s(j)} \in \cG$
contains the part with $\sigma^*(j)$. 
Similarly, for any $j \in \cX^o$ let $s(j)$ be any index such that $\sigma^*(j), \sigma(\kappa(j)) \in G_{s(j)}$
(note that since $\cX^o\cap\cX^{o^*}=\emptyset$, if $j\in\cX^o$ then $j\in\cX^{a^*}$).
Lemma \ref{lem:outliers} still holds.
For each $G_s\in\cG$, since $|G_s\cap\loc|,|G_s\cap\opt|\leq\rho'$ and $\loc$ is a local optimum, any test swap based on 
a group $G_s$ is not improving, hence $0 \leq \cost((\loc \triangle \loc_s) \cup \opt_s) - \cost(\loc)$.

For each test swap $G_s$, 
we describe how we could re-assign the each point $j$ for which $\sigma(j)$ becomes closed or 
and bound the cost of each re-assignment depending 
on the type of $j$. This case analysis is essentially the same as the one we had for \uuflo so it is deferred to Appendix \ref{app:kmed}.
Note that the points in $\cX^o\cup \cX^{o^*}$ are paired via $\kappa$. 


\section{Extensions}\label{sec:extensions}


\subsection{Extension to $\ell^q_q$-Norms}\label{sec:lp}
In this section we show how the results for \uuflo and \kmedo can be extended to the setting
where distances are $\ell^q_q$-norm for any $q\geq 1$. For example, if we have $\ell^2_2$-norm
distances instead of $\ell^1_1$ we get \kmeans (instead of \kmed). Let us consider modifying the analysis of \kmedo
to the $\ell^q_q$-norm distances.
In this setting for any local and global solutions $\loc$ and $\opt$, let $\delta_j=\delta(j,\sigma(j))$ and 
$\delta^*_j=\delta(j,\sigma^*(j))$; then $c_j=\delta^q_j$ and $c^*_j=\delta^{*q}_j$. 
It is easy to see that throughout the analysis of \uuflo and \kmedo, for the cases
that $j\in\cX^a\cap\cX^{a^*}$ is lucky, long, or good (but not bad) 
when we consider a test swap for a group $P$
we can bound the distance from $j$ to some point in $\loc \triangle P$ by first moving it to either $\sigma(j)$ or $\sigma^*(j)$
and then moving it a distance of $O(\eps) \cdot (\del_j + \del^*_j)$ to reach an open centre. Considering 
that we reassigned $j$ from $\sigma(j)$, the reassignment cost will be at most 

\[ (\del_j + O(\eps) \cdot (\del_j + \del^*_j))^q - c_j = O(2^q \eps) \cdot (c_j + c^*_j)\]
or
\[ (\del^*_j + O(\eps) \cdot (\del_j + \del^*_j))^q - c_j = (1 + O(2^q\eps)) \cdot c^*_j - (1 - O(2^q\eps)) \cdot c_j. \]

For the case that $j$ is bad, the reassignment cost can be bounded by $O(2^q(c_j+c^*_j))$ but since the probability of being
bad is still bounded by $\epsilon$, the bound in Equation (\ref{eqn:finalDelta}) can be re-written as:

\begin{eqnarray}
0 \leq \quad \sum_{\begin{array}{cc}
j \in \cX^{a} \cap \cX^{a^*} \\
j\text{ is not bad}
\end{array}} \left[(1+O(2^q \epsilon)) \cdot c^*_j - (1-O(2^q\epsilon)) c_j\right]
 + \sum_{j\text{ bad}}O(2^q (c^*_j + c_j)) + \sum_{j \in \cX^o} c^*_j - \sum_{j \in \cX^{o^*}} c_j,
\end{eqnarray}
which would imply 
\[ 0 \leq \sum_{j \in \cX} \left[(1+O(2^q\epsilon)) \cdot c^*_j - (1-O(2^q\epsilon)) \cdot c_j\right], \]
and  $\cost(\loc) \leq (1+O(2^q\epsilon)) \cdot OPT.$ It is enought to choose $\epsilon$ sufficiently small compared to $q$
to obtain a $(1+\epsilon')$-approximation.
Similar argument show that for the case of \uuflo with $\ell^q_q$-norm distances, we can get a PTAS.


\subsection{Minor-Closed Families of Graphs}

We consider the problem \kmedo in families of graphs that exclude a fixed minor $H$. Recall that a family of graphs is
 closed under minors
if and only if all graphs in that family exclude some fixed minor. 

Let $G = (V,E)$ be an edge-weighted graph excluding $H$ as a minor where $\cX, \fa \subseteq V$ and let $\delta$ denote the shortest-path metric of $G$.
We will argue Algorithm \ref{alg2:local} for some appropriate constant $\rho' := \rho'(\epsilon, H)$ returns a set 
$\loc \subseteq \fa$ with $|\loc| = (1+\epsilon) \cdot k$
where $\cost(\loc) \leq (1+\epsilon) \cdot OPT$. This can be readily adapted to \kmeanso using ideas from Section 
\ref{sec:lp}. We focus on \kmedo for simplicity. This will
also complete the proof of Theorem \ref{theo:kmeanso} for minor-closed metrics. The proof of Theorem \ref{theo:uflo} 
for minor-closed metrics is proven similarly and is slightly simpler.

We use the same notation as our analysis for \kmedo in doubling metrics. Namely, $\loc \subseteq \fa$ is a local optimum 
solution, $\opt \subseteq \fa$ is a global optimum solution, $\chi^a$ are the points
assigned in the local optimum solution, $\chi^o$ are the outliers in the local optimum solution, etc. We assume
 $\loc \cap \opt = \emptyset$ (one can ``duplicate'' a vertex by adding a copy with a 0-cost edge to the original
and preserve the property of excluding $H$ as a minor) and $\chi^o \cap \chi^{o^*} = \emptyset$.

A key difference is that we do not start with a perfect partitioning of $\loc \cup \opt$, as we did with doubling metrics.
 Rather, we start with the $r$-divisions described in \cite{AKM16B} which provides ``regions'' which consist
of subsets of $\loc \cup \opt$ with limited overlap.
We present a brief summary, without proof, of their partitioning scheme and how it is used to analyze the multiswap 
heuristic for \ufl.
Note that their setting is slightly different in that they show local search provides a true PTAS for \kmed and \kmeans,
whereas we are demonstrating a bicriteria PTAS for \kmedo and \kmeanso. It is much easier to describe a PTAS using their 
framework if $(1+\epsilon)\cdot k$ centres are opened in the algorithm.
Also, as the locality gap examples in the next section show, Algorithm \ref{alg2:local} may not be a good approximation 
when using solutions $\loc$ of size exactly $k$.

First, the nodes in $V$ are partitioned according to their nearest centre in $\loc \cup \opt$, breaking ties in a consistent manner. Each part (i.e. Voronoi cell) is then a connected component so each can be contracted to get a graph $G'$ with vertices 
$\loc \cup \opt$. Note $G'$ also excludes $H$ as a minor.
Then for $r = d_H/\epsilon^2$ where $d_H$ is a constant depending only on $H$, they consider an $r$-division of $G'$. 
Namely, they consider ``regions'' $R_1, \ldots, R_m \subseteq \loc \cup \opt$ with the following properties
(Definition III.1.1 in \cite{AKM16B})
First, define the boundary $\partial(R_a)$ for each region to be all centres $i \in R_a$ incident to an edge $(i,i')$ 
of $G'$ with $i' \not\in R_a$.
\begin{itemize}
\item Each edge of $G'$ has both endpoints in exactly one region.
\item There are at most $c_H/r \cdot (|\loc| + |\opt|)$ regions where $c_H$ is a constant depending only on $H$.
\item Each region has at most $r$ vertices.
\item $\sum_{a=1}^m |\partial(R_a)| \leq \epsilon \cdot (|\loc| + |\opt|)$.
\end{itemize}
In general, the regions are not vertex-disjoint.

For each region $R_a$, the test swap $\loc \rightarrow \loc - ((R_a - \partial(R_a)) \cap \loc) \cup (R_a \cap \opt)$ is considered. Each $j$ with $\sigma(j) \in R_a-\partial(R_a)$ is moved in one of two ways:
\begin{itemize}
\item If $\sigma^*(j) \in R_a$, move $j$ to $\sigma^*(j)$ for a cost change of $c^*_j-c_j$.
\item Otherwise, if point $j$ is in the Voronoi cell for some $i \in R_a$ then $\del(j, \partial(R_a)) \leq c^*_j$ because the shortest path from $j$ to $\sigma^*(j)$ must include a vertex $v$ in the Voronoi cell of some $i \in \partial(R_a)$.
By definition of the Voronoi partitioning, $i$ is closer to $v$ than $\sigma^*(j)$. So the cost change for $j$ is at most $c^*_j - c_j$ again.
\item Finally, if point $j$ does not lie in the Voronoi cell for any $i \in R_a \cup \partial(R_a)$ then $\del(j, \partial(R_a)) \leq c_j$ because the shortest path from $j$ to $\sigma(j)$ again crosses the boundary of $R_a$. So the cost change for $j$ is at most 0.
\end{itemize}
Lastly, for each $j \in \cX$ if no bound of the form $c^*_j - c_j$ is generated for $j$ according to the above rules, then $j$ should move from $\sigma(j)$ to $\sigma^*(j)$ in some swap that opens $\sigma^*(j)$.

We use this approach as our starting point for \kmedo. Let $\epsilon' > 0$ be a constant such that we run Algorithm \ref{alg2:local} using $(1+\epsilon') \cdot k$ centres in $\loc$. We will fix the constant $\rho'$ dictating the size
of the neighbourhood to be searched momentarily.

Form the same graph $G'$ obtained by contracting the Voronoi diagram of $G'$, 
let $R_1, \ldots, R_m$ be the subsets with the same properties as listed above for $\epsilon := \epsilon'/10$.
The following can be proven in a similar manner to Theorem \ref{theo:newpartition}. Its proof is deferred to Appendix \ref{app:newpartition2}.
\begin{lemma}\label{lem:newpartition2}
We can find regions $R'_1, \ldots, R'_\ell$ such that the following hold.
\begin{itemize}
\item Each edge lies in exactly one region.
\item $|R'_a| \leq O(r)$ for each $1 \leq a \leq \ell$.
\item $|(R'_a-\partial(R'_a)) \cap \loc| > |R'_a \cap \opt|$.
\end{itemize}
\end{lemma}

The rest of the proof proceeds as with our analysis of \kmedo in doubling metrics. For each $j \in \cX^o$ let $\tau(j)$ be any index such that $\sigma^*(j) \in R'_a$ and for each $j \in \cX^{o^*}$ let $\tau^*(j)$ be any index
such that $\sigma(j) \in R_a$. For each $R'_b$, define the imbalance of outliers $\Delta_b$ to be $|\{j \in \cX^o : \tau(j) = b\}| - |\{j \in \cX^{o^*} : \tau^*(j) = b\}|$. Note $\sum_{b=1}^\ell \Delta_b = 0$.

Everything now proceeds as before so we only briefly summarize: we find groups $G_s$ each consisting of $\Theta(r^2/\epsilon')$ regions of the form $R'_b$ along with similar a similar bijection $\kappa : \cX^o \rightarrow \cX^{o^*}$
such that for each $j \in \cX^o$ we have $R'_{\tau(j)}$ and $R'_{\tau^*(\kappa(j))}$ appearing in a common group. Finally, at most two regions are split by this grouping.
At this point, we determine that swaps of size $\rho' = O(r^3/\epsilon') = O(d^3_H / {\epsilon'}^7)$ in Algorithm \ref{alg2:local} suffice for the analysis.

For each group $G_s$, consider the test swap that opens all global optimum centres lying in some region $R'_a$ used to form $G_s$ and closes all local optimum centres in $G_s$ that are neither in the boundary $\partial(R'_a)$ 
of any region forming $G_s$ or in a split region. Outliers are reassigned exactly as with \kmedo, and points in $\cX^a \cap \cX^{a^*}$ are moved as they would be in the analysis for \ufl described above.

For each $j \in \cX^a \cap \cX^{a^*}$, the analysis above shows the total cost change for $j$ can be bounded by $c^*_j - c_j$. Similarly, for each $j \in \cX^o$ we bound the total
cost change of both $j$ and $\kappa(j)$ together by $c^*_j - c_{\kappa(j)}$. So in fact $\cost(\loc) \leq \cost(\opt)$.


\section{General Metrics}\label{sec:general}
Here we prove Theorem \ref{theo:general}. That is,
we show how to apply our framework for \kmedo and \kmeanso to the local search analysis in \cite{GT08} for \kmed and \kmeans
in general metrics where no assumptions are made about the distance function $\del$ apart from the metric properties. The algorithm and the redirections of the clients are the same with both \kmedo and \kmeanso.
We describe the analysis and summarize the different bounds at the end to handle outliers.

Let $\epsilon > 0$ be a constant and suppose Algorithm \ref{alg2:local} is run using solutions $\loc$ with $|\loc| \leq (1+\epsilon) \cdot k$
and neighbourhood size $\rho'$ for some large constant $\rho' > 0$ to be determined.
We use the same notation as before and, as always, assume $\loc \cap \opt = \emptyset$ and $\cX^o \cap \cX^{o^*} = \emptyset$.

Let $\phi : \opt \rightarrow \loc$ map each $i^* \in \opt$ to its nearest centre in $\loc$. Using a trivial adaptation of Algorithm 1 in \cite{GT08} (the only difference being we have $|\loc| = (1+\epsilon) \cdot k$
rather than $|\loc| = k$), we find {\em blocks} $B_1, \ldots, B_m \subseteq \loc \cup \opt$ with the following properties.
\begin{itemize}
\item The blocks are disjoint and each $i^* \in \opt$ lies in some block.
\item For each block $B_a$, $|B_a \cap \loc| = |B_a \cap \opt|$.
\item For each block $B_a$, there is exactly one $i \in B_a \cap \loc$ with $\phi^{-1}(i) \neq \emptyset$. For this $i$ we have $B_a \cap \opt = \phi^{-1}(i)$.
\end{itemize}
Call a block {\bf small} if $|B_a \cap \loc| \leq 2/\epsilon$, otherwise it is {\bf large}. Note there are at most $\frac{\epsilon}{2} \cdot k$ large blocks. On the other hand, there are $\epsilon \cdot k$ centres in $\loc$
that do not appear in any block. Assign one such unused centre to each large block, note these centres $i$ satisfy $\phi^{-1}(i) = \emptyset$ and there are still at least $\frac{\epsilon}{2} \cdot k$ centres in $\loc$ not appearing in any block.

Now we create parts $P_s$. For each large block $B_a$, consider any paring between $B_a \cap \opt$ and
$\{i \in B_a \cap \loc : \phi^{-1}(i) = \emptyset\}$. Each pair forms a part on its own. Finally, each small block is a part on its own. This is illustrated in Figure \ref{fig:general}.

\begin{figure}[ht]
\begin{center}
\includegraphics[scale=0.53]{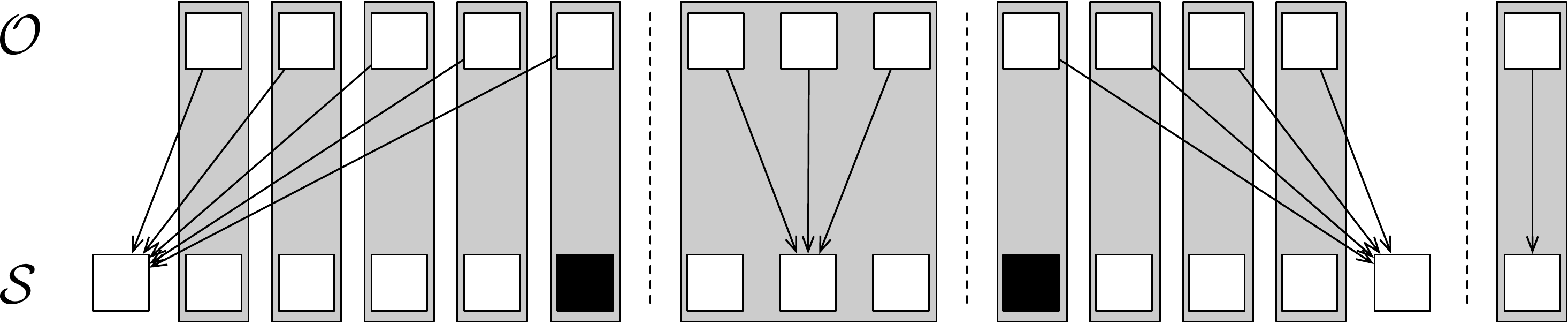}
\end{center}
\caption{Depiction of the formation of the parts $P_s$ where the arrows between them indicate the mapping $\phi$.
The white nodes are in the initial blocks and the blocks are separated by white dashed lines. The first and third block are large, so another centre (black) not yet in any block is added to them.
The grey rectangles are the parts that are formed from the blocks.}\label{fig:general}
\end{figure}

Note each part $P_s$ satisfies $|P_s \cap \loc| = |P_s \cap \opt| \leq 2/\epsilon$. Then perform the following procedure: while there are at least two parts with size at most $2/\epsilon$, merge them into a larger part.
If there is one remaining part with size at most $2/\epsilon$, then merge it with any part created so far. Now all parts have size between $2/\epsilon$ and $6/\epsilon$, call these groups $P'_1, \ldots, P'_\ell$.
Note $\ell \leq \frac{\epsilon}{2} \cdot k$ because the sets $P'_a \cap \opt$ partition $\opt$. To each part, add one more $i \in \loc$ which does not yet appear in a part. Summarizing,
the parts $P'_1, \ldots, P'_\ell$ have the following properties.
\begin{itemize}
\item Each $i^* \in \opt$ appears in precisely one part. Each $i \in \loc$ appears in at most one part.
\item $|P'_a \cap \opt| < |P'_a \cap \loc| \leq \frac{6}{\epsilon} + 1 \leq \frac{7}{\epsilon}$.
\end{itemize}

Construct a pairing $\kappa : \cX^o \rightarrow \cX^{o^*}$ almost as before. First, pair up outliers within a group arbitrarily. Then arbitrarily pair up unpaired $j \in \cX^o$ with points $j' \in \cX^{o^*}$ such that $\sigma(j)$
does not appear in any part $P'_a$. Finally, pair up the remaining unpaired outliers using Algorithm \ref{alg:pairing} applied to these parts. Form groups $G_s$ with these parts in the same way as before. Again, note some $i \in \loc$
do not appear in any group, this is not important. What is important is that each $i^* \in \opt$ appears in some group.

The swaps used in the analysis are of the following form. For each group $G_s$ we swap in all global optimum centres appearing in $G_s$ and swap out all local optimum centres appearing in $G_s$ that are not in a part split by $G_s$.
As each group is the union of $\Theta(1/\epsilon)$ parts, the number of centres swapped is $O(\epsilon^{-2})$. This determines $\rho' := O(\epsilon^{-2})$ in Algorithm \ref{alg2:local} for the case of general metrics.

The clients are reassigned as follows. Note by construction that if $j \in \cX^{o^*}$ has $\sigma(j)$ in a part that is not split then $\sigma^*(\kappa^{-1}(j))$ lies in the same group as $\sigma(j)$. Say that this is the group when $\kappa^{-1}(j)$ should be connected.
Otherwise, pick any group containing $\sigma^*(\kappa^{-1}(j))$ and say this is when $\kappa^{-1}(j)$ should be connected.

For each group $G_s$,
\begin{itemize}
\item For each $j \in \cX^o$, if $j$ should be connected in this group then connect $j$ to $\sigma^*(j)$ and make $\kappa(j)$ an outlier for a cost change of $c^*_j - c_{\kappa(j)}$.
\item For any $j \in \cX^a \cap \cX^{a^*}$ where $\sigma(j)$ is closed, move $j$ as follows:
\begin{itemize}
\item If $\sigma^*(j)$ is now open, move $j$ to $\sigma^*(j)$ for a cost change bound of $c^*_j - c_j$.
\item Otherwise, move $j$ to $\phi(\sigma^*(j))$ which is guaranteed to be open by how we constructed the parts. The cost change here is summarized below for the different cases of \kmed and \kmeans.
\end{itemize}
\end{itemize}
Finally, for any $j \in \cX^a \cap \cX^{a^*}$ that has not had its $c^*_j - c_j$ bound generated yet, move $j$ to $\sigma^*(j)$ in any one swap where $\sigma^*(j)$ is opened to get a bound of $c^*_j - c_j$ on its movement cost.

The analysis of the cost changes follows directly from the analysis in \cite{GT08}, so we simply summarize.
\begin{itemize}
\item For \kmedo, $\cost(\loc) \leq (3+O(\epsilon)) \cdot OPT$.
\item For \kmeanso, $\cost(\loc) \leq (25 + O(\epsilon)) \cdot OPT$.
\item For $k$-clustering with $\ell_q^q$-norms of the distances, $\cost(\loc) \leq ((3 + O(\epsilon)) \cdot q)^q$.
\end{itemize}
These are slightly different than the values reported in \cite{GT08} because they are considering the $\ell_q$ norm whereas we are considering $\ell_q^q$.
This concludes the proof of Theorem \ref{theo:general}.



\section{The Locality Gap}\label{sec:gap}

In this section, we show that the natural local search heuristic has unbounded gap for \uflo (with non-uniform opening costs).
We also show that local search multiswaps that do not violate the number of clusters and outliers can 
have arbitrarily large locality gap for \kmedo and \kmeans, even in the Euclidean metrics. 
Locality gap here refers to the ratio of any local optimum solution produced by the local 
search heuristics to the global optimum solution. 


\subsection{ \uflo with Non-uniform Opening Costs}

First, we consider a multi-swap local search heuristic for the \uflo problem with non-uniform centre opening costs. We 
show this for any local search that does $\rho$-swaps and does not discard more than $z$ 
points as outliers, where $\rho$ is a constant and $z$ is part of the input has unbounded ratio.
Assume the set of $\cX \cup \cl$ is partitioned into disjoint sets $A, \, B_1, \, B_2, \, \ldots, \, B_z$, where:
\begin{itemize}
	\item The points in different sets are at a large distance from one another.
	\item $A$ has one centre $i$ with the cost of $\rho$ and $z$ points which are colocated at $i$.
	\item For each of $\ell = 1, \, 2, \, \ldots, \, z$, the set $B_\ell$ has one centre $i_\ell$ with the cost of $1$ and one point located at $i_\ell$.
\end{itemize}

The set of centres $\calS = \{i_1, \, i_2, \, \ldots, \, i_z\}$ is a local optimum for the $\rho$-swap local search heuristic; any $\rho$-swap between $B_\ell$'s will not reduce the cost, and any $\rho$-swaps that 
opens $i$ will incur an opening cost of $\rho$ which is already as expensive as the potential savings from closing $\rho$ of the $B_\ell$'s. Note that the $z$ points of $A$ are discarded as outliers in $\calS$. 
It is straightforward to verify that the global optimum in this case is $\opt = \{ i \}$, which discards the points in $B_1, \, \ldots, \, B_z$'s as outliers. The locality gap for this instance is $\cost(\calS) / \cost(\opt)
= \dfrac{z}{\rho}$, which can be arbitrarily large for any fixed $\rho$.

Note this can also be viewed as a planar metric, showing local search has an unbounded locality gap for \uflo in planar graphs.

\subsection{\kmedo and \kmeanso }

Chen presents in \cite{Chen07} a bad gap example for local search for \kmedo in general metrics. 
The example shows an unbounded locality gap for the multiswap local search heuristic that does not violate the 
number of clusters and outliers. We adapt this example to Euclidean metrics. The same example shows standard multiswap local search for \kmeanso has an unbounded locality gap.

Consider an input in which $n  \gg z \gg  k > 1$. The set of points $\cX$ is partitioned into disjoint sets $B$, $C_1, \, C_2, \, \ldots, \, C_{k - 1}$, $D_1, \, D_2, \, \ldots, \, D_{k - 2}$, and $E$:
\begin{itemize}
	\item The distance between every pair of points from different sets is large.
	\item $B$ has $n - 2z$ colocated points.
	\item For each $i = 1, \, 2, \, \ldots, \, k - 1$, $C_i$ has one point in the centre and $u - 1$ points 
	evenly distributed on the perimeter of a circle with radius $\beta$ from the centre.
	\item For each $j = 1, \, 2, \, \ldots, \, k - 2$, $D_j$ has $u - 1$ colocated points.
	\item $E$ has one point at the centre and $u + k - 3$ points evenly distributed on the perimeter of the circle with 
        radius $\gamma$,
\end{itemize}
where $u = z / (k - 1)$, and $\beta$ and $\gamma$ are chosen such that $\gamma < (u - 1) \beta <  2\gamma$ (see Figure \ref{fig:kmedian_locality}). 
\begin{figure}[ht]
\begin{center}
\includegraphics[scale=0.53]{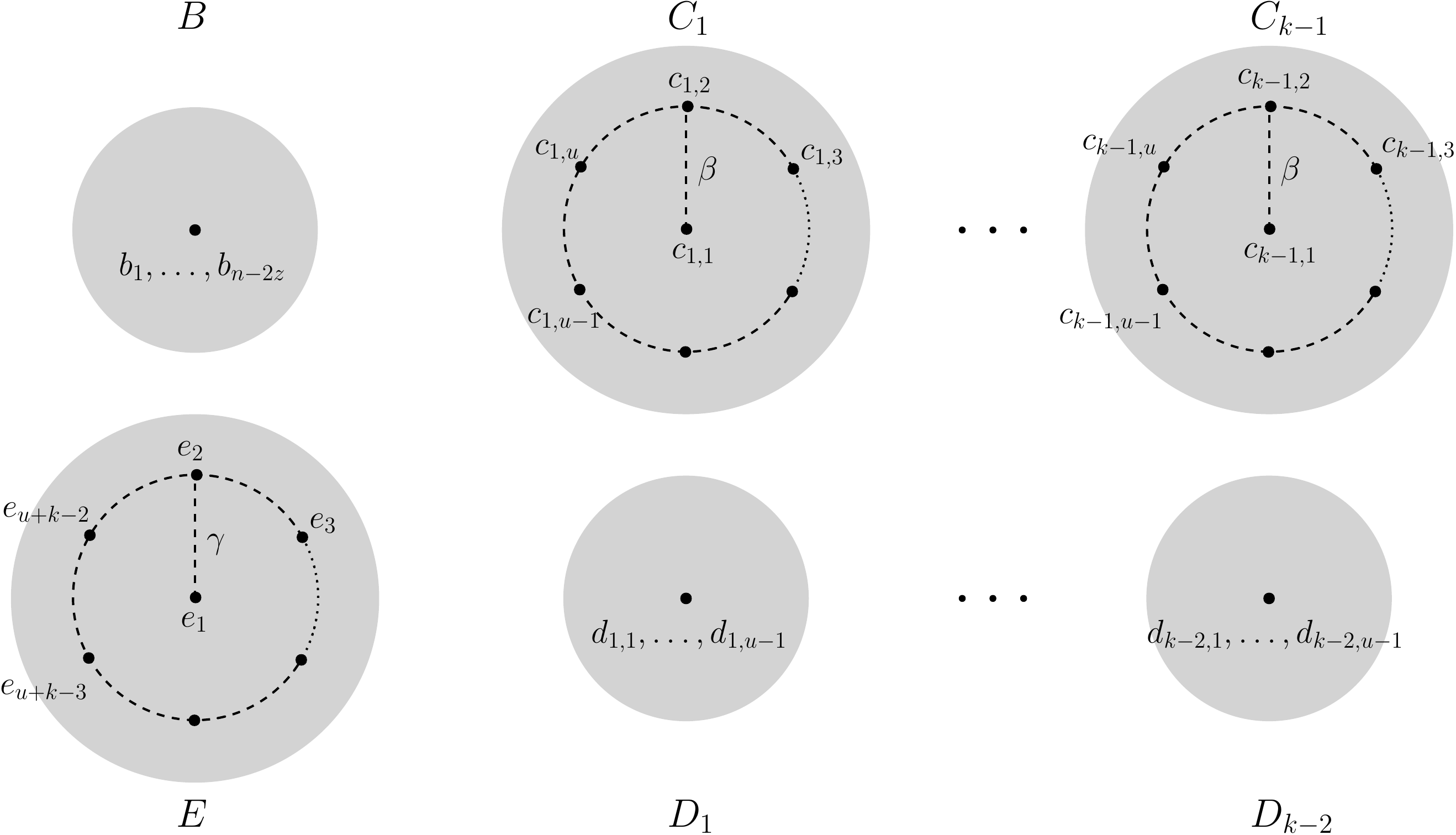}
\end{center}
\caption{The locality gap counter example for \kmed in Euclidean metrics.}\label{fig:kmedian_locality}
\end{figure}
Let $f(.)$ denote the centre point of a set (in the case of colocated sets, any point from the set). Then, the set 
$\calS = \{f(B), \, f(D_1), \, f(D_2), \, \ldots, \, f(D_{k - 2}), \, f(E) \}$ is a local optimum for the 
$\rho$-swap local search if $\rho < k - 1$ with cost $(u+k-3)\gamma$. 
The reason is that since the distance between the sets is large, we would incur a large cost by closing 
$f(B)$ or $f(E)$. Therefore, we need to close some points in the sets $D_1, \, \ldots, \, D_{k - 2}$, and
 open some points in $C_1, \, \ldots, \, C_{k - 1}$ to ensure we do not violate the number of outliers.
 Since $z \gg k$, we can assume $u > k - 1$. We can show via some straightforward algebra that if we
 close $\rho$ points from $D_j$'s, then we need to open points from exactly $\rho$ different $C_i$'s to
 keep the number of outliers below $z$. Since the points on the perimeter are distributed evenly, we
 incur the minimum cost by opening $f(C_i)$'s. So, we only need to show that swapping at most $\rho$ 
 points from $f(D_j)$'s with $\rho$ points in $f(C_i)$'s does not reduce the cost. Assume w.l.o.g. that 
 we swap $f(D_1), \, f(D_2), \, \ldots, \, f(D_\rho)$ with $f(C_1), \, f(C_2), \, \ldots, \, f(C_\rho)$. The
 new cost is $\rho (u - 1) \beta + (u + k - \rho - 3) \gamma$. Notice that $\cost(\calS) = (u + k - 3) 
 \gamma$. Therefore, the cost has, in fact, increased as a result of the $\rho$-swap since $(u - 1)
 \beta > \gamma$. Now, we can show the following claim for \kmedo:
 \begin{claim}
 The $\rho$-swap local search heuristic that generates no more than $k$ clusters and discards at most $z$ outliers has an unbounded locality gap in Euclidean metrics, if $\rho < k - 1$.
 \end{claim}
 
 Consider the solution $\opt = \{ f(B), \, f(C_1), \, f(C_2), \, \ldots, \, f(C_{k - 1}) \}$ which costs $(k - 1)
 (u - 1) \beta$. This is indeed the global optimum for the given instance. The locality gap then would be 
 \[\dfrac{(u - k + 3) \gamma}{(k - 1)(u - 1) \beta} > \dfrac{u - k + 2}{2(k - 1)},\]
since $(u - 1) \beta < 2 \gamma$. This ratio can be arbitrarily large as $z$ (and consequently $u$) grows.

A slight modification of this example to planar metrics also shows local search has an unbounded locality gap for \kmedo and \kmeanso in planar graphs. In particular, the sets of collocated points
can be viewed as stars with 0-cost edges and the sets $E$ and all $C_i$ can be viewed as stars where the leaves are in $\cX$ and have distance $\gamma$ or $\beta$ (respectively) to the middle of the star, which lies in $\fa$.


\bibliographystyle{plain}
{
\bibliography{k-means-outliers}
}

\appendix

\section{Missing proofs for Section \ref{sec:uflo}}\label{app:missingProofs}
\begin{pproof}{Lemma \ref{lem:size}}
The upper bound is simply because the number of parts used to form $G_s$ is at most $2 |E_s| \leq 4\alpha$ and each part has at most $2\rho$ centres.
For the lower bound, we argue the graph $(\pi^+ \cup \pi^-, \cP)$ is acyclic. If so, then $(G_s , E_s)$ is also acyclic
and the result follows because any acyclic graph $(V,E)$ has $|V| \geq |E|-1$ (and also the fact that $|P| \geq 1$ for any part).

To show $(\pi^+ \cup \pi^-, \cP)$ is acyclic, we first remove all nodes (and incident edges) with degree 1. Call this graph $H$. 
After this, the maximum degree of a vertex/part in $H$ is at least 2. To see this, note for any $P \in \pi^+ \cup \pi^-$,
the incident edges are indexed consecutively by how the algorithm constructed super edges. 
If $e_i, e_{i+1}, \ldots, e_j$ denotes these super edges incident to $P$, then, again by how the algorithm constructed super edges,
for any $i < c < j$ that the other endpoint of $e_c$ only appeared in one iteration so the corresponding part has degree 1.

Finally, consider any simple path in $H$ with at least 4 vertices starting in $\pi^+$ and ending in $\pi^-$. Let
$P^+_{a_1}, P^-_{b_1}, P^+_{a_2}, P^-_{b_2}, \ldots, P^+_{a_c}, P^-_{b_c}$ be, in order, the parts visited by the path.
Then both the $a_i$ and $b_i$ sequences are strictly increasing, or strictly decreasing. As $c \geq 2$
(because there are at least 4 vertices on the path) then either i) $a_1 < a_c$ and $b_1 < b_c$ or ii) $a_1 > a_c$ and $b_1 > b_c$. 
Suppose, without loss of generality, it is the former case.
There is no edge of the form $(P^+_{a_1}, P^-_{b_c})$ in $\cP$ because the only other edge $(P^+_a, P^-_{b_c})$ incident to 
$P^-_{b_c}$ besides $(P^+_{a_c}, P^-_{b_c})$ has $a > a_c$.

Ultimately, this shows there is no cycle in $H$. As $H$ is obtained by removing only the degree-1 vertices of 
$(\pi^+ \cup \pi^-, \cP)$, then this graph is also acyclic.
\end{pproof}

\begin{pproof}{Lemma \ref{lem:split}}
Say $E_s = \{e_i, e_{i+1}, \ldots, e_j\}$. First consider some index $i < c < j$ where $e_c$ has a split endpoint $P$. Then some edge $e_{c'}$ incident to $P$ is not in $E_s$
so either $c' < i$ or $c' > j$. Suppose $c' < i$, the other case $c' > j$ is similar.
By construction, all edges in $N_{\cP}(P)$ for any $P \in \pi^+\cup \pi^-$ appear consecutively. As $c' < i < c$ and $e_{c'}, e_c \in N_{\cP}(P)$, then $e_i \in N_{\cP}(P)$.
Thus, the only split parts are endpoints of either $e_i$ or $e_j$.

We claim that if both endpoints of $e_i$ are split, then $e_i$ and $e_j$ share an endpoint and the other endpoint of $e_j$ is not split.
Say $e_{i-1} = (P^+_{a_{i-1}}, P^-_{b_{i-1}})$. Then either $P^+_{a_{i-1}} \neq P^+_{a_i}$ or $P^-_{b_{i-1}} \neq P^-_{b_i}$. Suppose it is the former case, the latter again being similar.
So if $P^+_{a_i}$ is split it must be $e_{j+1}$ has $P^+_{a_i}$ as an endpoint. Therefore, $e_i$ and $e_j$ share $P^+_{a_i}$ as an endpoint. Further, since $e_{j+1}$ has $P^+_{a_i}$ as an endpoint
then $N_{\cP}(P^-_{b_j}) = \{e_j\}$ so $P^-_{b_j}$, the other endpoint of $e_j$, has degree 1 so is not split.

Similarly, if both endpoints of $e_j$ are split then one is in common with $e_i$ and the other endpoint of $e_i$ is not split. Overall we see $G_s$ splits at most two vertices.
\end{pproof}

\section{Proof of Theorem \ref{theo:newpartition}}\label{app:newpartition}

Here we show how we can modify proof of Theorem 4 in \cite{FRS16B} to prove Theorem \ref{theo:newpartition}.
The proof of Theorem 4 in \cite{FRS16B} starts by showing the existence of a randomized partitioning scheme of 
$\loc\cup\opt$ where each part has size at most $(d/\epsilon)^{\Theta(d/\epsilon)}$ satisfying the 2nd and 3rd condition.
We can ensure that each part has size at least $1/\epsilon$ by merging small parts if needed.
That part of the proof remains unchanged. Let us call the parts generated so far $P_1,\ldots,P_\ell$.
Then we have to show how we can combine constant number of parts to satisfy
condition 1 for some $\rho=\rho(\epsilon,d)$. Since we have $(1+\epsilon)k$ centres in $\loc$ and $k$ centres in $\opt$
we can simply add one ``dummy'' optimum centre to each part $P_i$ so that each part has now one dummy centre, 
noting that  $\ell< \epsilon k$ (because each part has size at least $1/\epsilon$). 
We then perform the balancing step of proof of Theorem 4 in \cite{FRS16B} to obtain parts of size 
$\rho=(d/\epsilon)^{\Theta(d/\epsilon)}$ with each part having the same number of centres from $\loc$ and $\opt$
satisfying conditions 2 and 3. Removing the ``dummy'' centres, we satisfy condition 1 of Theorem \ref{theo:newpartition}.


\section{Reassigning Points for the \kmedo Analysis} \label{app:kmed}
Below we specify what to do for each point $j\in\cX^o$ and $\kappa(j)$ together.

\begin{itemize}
\item If $j\in\cX^o$ and $s=s(j)$, then by Lemma \ref{lem:outliers} $\sigma^*(j)\in G_s$ and it is open, we
assign $j$ to $\sigma^*(j)$ and make $\kappa(j)$ an outlier. The total assignment cost change for
$j$ and $\kappa(j)$ will be $c^*_j-c_{\kappa(j)}$.

\end{itemize}

The subsequent cases are when $j\in\cX^{a}\cap\cX^{a^*}$.

\begin{itemize}
\item If $j$ is lucky and $s=s(j)$ then we reassign $j$ from $\sigma(j)$ to $\sigma^*(j)$. The total
reassignment cost change is $c^*_j-c_j$.

\item If $j$ is long then we assign $j$ to the nearest open center to $\sigma(j)$. Again using Lemma
\ref{lem:cents} and since $j$ is long, the total cost change of the reassignment is at most:

\[ 5 \cdot D_{\sigma(j)} \leq 5 \epsilon \cdot \delta(\sigma(j), \sigma^*(j)) \leq 5\epsilon \cdot (c^*_j + c_j). \]

\item If $j$ is good and $D_{\sigma^*(j)} \leq \epsilon \cdot D_{\sigma(j)}$, we assign $j$ to the open centre that 
is nearest to $\sigma^*(j)$. By \eqref{eqn:radius} in Section \ref{sec:recall} and Lemma \ref{lem:cents},
the assignment cost change for $j$ is at most:
\begin{alignat*}{2}
c^*_j + 5 \cdot D_{\sigma^*(j)} - c_j \quad
& \leq \quad c^*_j + 5 \epsilon \cdot D_{\sigma(j)} - c_j \\
& \leq \quad c^*_j + 5\epsilon\cdot (c^*_j + c_j)  - c_j \\
& \leq \quad (1+5\epsilon) \cdot c^*_j - (1-5\epsilon)\cdot c_j.
\end{alignat*}
\item If $j$ is good but $D_{\sigma^*(j)} > \epsilon \cdot D_{\sigma(j)}$, then let $i'$ be such that $\sigma(j),i'$ both 
lie in $G_s$ and $\delta(\sigma^*(j), i') \leq \epsilon \cdot D_{\sigma^*(j)}$. Reassigning $j$ from $\sigma(j)$ to $i'$
bounds its assignment cost change by
\[ c^*_j + \delta(\sigma^*(j), i')  -c_j \leq c^*_j + \epsilon \cdot D_{\sigma^*(j)} - c_j 
\leq (1+\epsilon) \cdot c^*_j - (1-\epsilon) \cdot c_j. \]
\item Finally, if $j$ is bad then simply reassign $j$ to the open centre that is nearest to $\sigma(j)$. 
By \eqref{eqn:radius} and Lemma \ref{lem:cents}, the total reassignment cost change for $j$ is at most:
\[ 5 \cdot D_{\sigma(j)} \leq 5 \cdot (c^*_j + c_j). \]
\end{itemize}

For every point $j$ for which $\sigma(j)$ is still open we keep it assigned to $\sigma(j)$.
Considering all cases, if $\Delta_j$ denotes the net cost change for re-assignment of $j\in\cX$, then:

\begin{itemize}
\item If $j \in \cX^{o}$ then the only time $j$ is moved is for the swap involving $G_{s(j)}$ and $\Delta_j = c^*_j$.
\item If $j \in \cX^{o^*}$ then the only time $j$ is moved is for the swap involving $G_{s(\kappa^{-1}(j))}$. So $\Delta_j = -c_j$.
\item If $j$ is lucky then it is only moved when $G_{s(j)}$ is processed so $\Delta_j = c^*_j - c_j$.
\item If $j$ is long then it is moved to $\sigma^*(j)$ when $G_{s(j)}$ is processed and it is moved near $\sigma(j)$ when $\sigma(j)$ is closed, so
\[ \Delta_j \leq c^*_j - c_j + 5\epsilon \cdot (c^*_j + c_j) = (1+5\epsilon) \cdot c^*_j - (1-5\epsilon) \cdot c_j. \]
\item If $j$ is good then it is only moved when $\sigma(j)$ is closed so $\Delta_j \leq (1+5\epsilon) \cdot c^*_j - (1-5\epsilon) \cdot c_j$.
\item If $j$ is bad then it is only moved when $\sigma(j)$ is closed so $\Delta_j \leq 5 \cdot (c^*_j + c_j)$.
\end{itemize}

Considering all $G_s\in\cG$ (and considering all test-swaps):

\begin{eqnarray}
0 & \leq& \quad \sum_{G_s \in \cG} \cost((\loc - \loc_s) \cup \opt_s) - \cost(\loc) \nonumber\\
 &\leq& \sum_{j \in \cX} \Delta_j \nonumber\\
& \leq & \sum_{\begin{array}{cc}
j \in \cX^{a} \cap \cX^{a^*} \\
j\text{ is not bad}
\end{array}} \left[(1+5\epsilon) \cdot c^*_j - (1-5\epsilon) c_j\right]
+ \sum_{j\text{ bad}}5(c^*_j + c_j) + \sum_{j \in \cX^o} c^*_j - \sum_{j \in \cX^{o^*}} c_j .\label{eqn:finalDelta}
\end{eqnarray}

Using the fact that the probability of a point $j$ being bad is at most $\epsilon$ we get:

\[ 0 \leq \sum_{j \in \cX} \left[(1+10\epsilon) \cdot c^*_j - (1-10\epsilon) \cdot c_j\right]. \]
Rearranging and relaxing slightly further shows the same bound $\cost(\loc) \leq (1+30\epsilon) \cdot OPT.$


\section{Proof of Lemma \ref{lem:newpartition2}}\label{app:newpartition2}
This is similar to the proof in Appendix \ref{app:newpartition}. The main difference in this setting is that the sets $R_a$ are not necessarily disjoint but are only guaranteed have limited overlap.

First, note
\begin{alignat}{2}
\sum_{a=1}^m  |(R_a - \partial(R_a) \cap \loc| - |R_a \cap \opt| \quad
& \geq [|\loc| - \epsilon \cdot (|\loc| + |\opt|)] -  \quad [|\opt| + \epsilon\cdot (|\loc| + \opt|)] \nonumber \\
& = \quad |\loc| - |\opt| -  2\epsilon\cdot (|\loc| + |\opt|) \nonumber \\
& = \quad \epsilon'\cdot k - 2\epsilon \cdot (2+\epsilon') \cdot k \nonumber \\
& \geq \quad \frac{\epsilon'}{2} \cdot k \label{eqn:minorbound}
\end{alignat}
because $\sum_{a=1}^m |\partial(R_a)| \leq \epsilon \cdot |\loc \cup \opt|$.

Now add ``dummy'' optimum centres to each region to form regions $\overline{R_a}$ satisfying
$\sum_{a=1}^m |\overline R_a \cap \opt| - |(\overline R_a - \partial(\overline R_a) \cap \loc| = 0$
where the boundary $\partial(\overline R_a)$ is the boundary of the non-dummy vertices. The number to be added overall is at bounded as follows,
\begin{alignat*}{2}
\sum_{a=1}^m  |(R_a - \partial(R_a) \cap \loc| - |R_a \cap \opt| \quad
& \leq \quad |\loc| - |\opt| - \epsilon \cdot (|\loc| + \opt|) \\
& \leq \quad \epsilon' k - 2\epsilon k \\
& = \quad \frac{4\epsilon'}{5} \cdot k.
\end{alignat*}
Now, as $|R_a| \leq r$ for each $a$ then there are at least $2k/r$ regions, so we may add at most $\frac{4\epsilon}{5} / \frac{2}{r} = O(d_H / \epsilon)$ centres per region to achieve this.
That is, $|\overline R_a| \leq |R_a| + O(d_H/\epsilon) = O(r)$.

On the other hand, we can guarantee each $\overline R_a$ has at least one dummy centre. this is because there are at least $\frac{\eps'}{2} \cdot k$ dummy centres to add and the number of regions
is at most $c_H/r \cdot 3k = \frac{3c_H}{d_H}\epsilon^2 \cdot k = \frac{300 c_H}{d_H} {\epsilon'}^2 \cdot k$. For small enough $\epsilon'$, the number of centres to add is at least the number of regions.

Finally, using Theorem 4 in \cite{FRS16B}, we can partition $\{\overline R_1, \ldots, \overline R_a\}$ into parts $\mathcal{R_1}, \ldots, \mathcal R_{\ell}$
where each part $\mathcal R_b$ consists of $O(r^2)$ regions and $\sum_{\overline R_a \in \mathcal R_b} |(\overline R_a - \partial(\overline R_a) \cap \loc| - |\overline R_a \cap \opt| = 0$.
For each $1 \leq b \leq \ell$, let $R'_b = \cup_{\overline R_a \in \mathcal R_b} R_a$ (discarding the dummies). Note $N(R'_b) \subseteq \cup_{\overline R_a \in \mathcal R_b} N(\overline R_a)$.
So $|R'_b \cap \opt| < |(R'_b - N(R'_b)) \cap \loc|$, as required.

\end{document}